\newif\ifanonymous
\newif\ifdraft

\anonymousfalse
\draftfalse
\InputIfFileExists{localflags}{}

\documentclass[ 
    conference 
]{IEEEtran}

\usepackage{versions}

\excludeversion{conf}
\includeversion{full}
\includeversion{notextended}

\usepackage{amsthm} 
\usepackage{amsmath}
\usepackage{amssymb}
\usepackage{stmaryrd}
\usepackage{mathabx} 
\usepackage[utf8]{inputenc}
\usepackage{xspace}
\usepackage{proof}
\usepackage{color}
\usepackage{paralist}
\usepackage[inline]{enumitem}
\usepackage{thm-restate}
\usepackage{booktabs} 

\usepackage{hyperref}
\usepackage{msc}
\usepackage{xparse} 

\usepackage{listings} 
\usepackage{syntax} 

\usepackage[
    backend=biber,
    sortlocale=en_GB,
    natbib=true,
    url=true, 
    doi=false,
    eprint=false
]{biblatex} 
\usepackage{tikz}
\usetikzlibrary{tikzmark}
\usetikzlibrary{decorations.pathreplacing,calc}

\AtEveryBibitem{
 \clearlist{address}
 \clearfield{date}
 \clearfield{eprint}
 \clearfield{isbn}
 \clearfield{issn}
 \clearlist{location}
 \clearfield{month}
 \ifentrytype{book}{
  \clearname{editor}
  }{}
\ifentrytype{inproceedings}{
  \clearname{editor}
  \clearfield{volume}
  \clearfield{url}
  \clearfield{pages}
  \clearlist{publisher}
  \clearfield{series}
  \clearlist{organization}
}{}
\ifentrytype{incollection}{
  \clearname{editor}
  \clearfield{volume}
  \clearfield{url}
  \clearlist{publisher}
  \clearlist{organization}
}{}
\ifentrytype{proceedings}{
 \clearname{editor}
 \clearfield{volume}
 \clearfield{url}
}{}
\ifentrytype{article}{
 \clearfield{url}
 \clearfield{pages}
}{}
}

\lstdefinelanguage{apip}{
  morekeywords=[1]{
	out, in, if, then, else, event, insert, delete, lookup, as, in, lock, unlock, let
    },
  sensitive=true,
  morecomment=[l]{\#},
  morecomment=[l]{//},
  morestring=[b]',
  morestring=[s]{`}{'},
}

\lstdefinestyle{apip}{
    language={apip}, 
	stringstyle=\rm\textup,
	keywordstyle=[1]\sf,
 	keywordstyle=[2]\rm\upshape,
	literate={||}{{$\mid$}}1
	{new}{{$\nu$}}1
	{<}{{$\langle$}}1
	{>}{{$\rangle$}}1
	{(}{{$($}}1
	{)}{{$)$}}1
	{:=}{{$\defeq$}}2
}

\makeatletter
\DeclareRobustCommand{\defeq}{\mathrel{\rlap{%
  \raisebox{0.3ex}{$\m@th\cdot$}}%
  \raisebox{-0.3ex}{$\m@th\cdot$}}%
  =}
\DeclareRobustCommand{\eqdef}{=\mathrel{\rlap{%
  \raisebox{0.3ex}{$\m@th\cdot$}}%
  \raisebox{-0.3ex}{$\m@th\cdot$}}%
  }
\makeatother

\newcommand{\set}[1]{\{#1\}}

\newtheorem{definition}{Definition}

\newtheorem{lemma}{Lemma}

\newtheorem{theorem}{Theorem}
\newtheorem{example}{Example}

\newtheoremstyle{mycase}{}{}{}{}{\bf}{.}{.5em}{\thmnote{#1:}~\normalfont
  #3}
\theoremstyle{mycase}

\numberwithin{subcase}{mycase}

\newtheoremstyle{component}{}{}{}{}{\itshape}{.}{.5em}{\thmnote{#3}#1}
\theoremstyle{component}

\DeclareMathOperator*{\dom}{\Dom}

\newcommand{\ledot}{\mathrel{\ooalign{\hss\raise.200ex\hbox{$\cdot$}\hss\cr$\le$}}}
\newcommand{\gedot}{\mathrel{\ooalign{\hss\raise.200ex\hbox{$\cdot$}\hss\cr$\ge$}}}

\newcommand\setN{\mathbb N}

\newcommand{\mcup}{\ensuremath{\cup^{\#}\xspace}}

\newcommand{\calF}{\ensuremath{\mathcal{F}}\xspace}

\newcommand{\calL}{\ensuremath{\mathcal{L}}\xspace}
\newcommand{\calM}{\ensuremath{\mathcal{M}}\xspace}
\newcommand\calP{\ensuremath{\mathcal{P}}\xspace}

\newcommand\calS{\ensuremath{\mathcal{S}}\xspace}

\newcommand\calX{\ensuremath{\mathcal{X}}\xspace}

\newcommand{\sig}{\ensuremath{\mathit{sig}}\xspace}

\newcommand{\pk}{\ensuremath{\mathit{pk}}\xspace}
\newcommand{\sk}{\ensuremath{\mathit{sk}}\xspace}

\newcommand{\pair}[2]{\ensuremath{\langle #1, #2 \rangle}}
\newcommand{\fst}{\ensuremath{\mathsf{fst}}}
\newcommand{\snd}{\ensuremath{\mathsf{snd}}}
\newcommand{\etal}{\textit{et al.}}

\newcommand{\mathcmd}[1]{{\normalfont\ensuremath{#1}}\xspace}

\newcommand{\mathname}[1]{\mathcmd{\text{\textrm{#1}}}}

\newcommand{\mathfun}[1]{\mathcmd{\mathit{#1}}}

\newcommand{\mathvalue}[1]{\mathcmd{\mathit{#1}}}

\newcommand{\mathset}[1]{\mathname{#1}}

\newcommand{\textop}[1]{\relax\ifmmode\mathop{\text{#1}}\else\text{#1}\fi}
 
\newcommand{\trans}[1]{\ensuremath{{\xrightarrow{#1}}}\xspace}
\newcommand{\pitrans}[1]{\ensuremath{\stackrel{#1}{\Longrightarrow}}\xspace}

\newcommand{\msgsort}{\ensuremath{\mathit{msg}}\xspace}
\newcommand{\freshsort}{\ensuremath{\mathit{fresh}}\xspace}
\newcommand{\pubsort}{\ensuremath{\mathit{pub}}\xspace}
\newcommand{\tempsort}{\ensuremath{\mathit{temp}}\xspace}

\newcommand{\PN}{\ensuremath{\mathit{PN}}\xspace}
\newcommand{\FN}{\ensuremath{\mathit{FN}}\xspace}
\newcommand{\Vars}{\ensuremath{\mathcal{V}}\xspace}

\newcommand{\Sign}{\ensuremath{\Sigma}\xspace}
\newcommand{\ET}{\ensuremath{E}\xspace}
\newcommand{\Terms}{\ensuremath{\mathset{Terms}}\xspace}
\newcommand{\Mess}{\ensuremath{\mathcal{M}}\xspace}
\newcommand{\FSign}{\ensuremath{\Sigma_\mathit{fact}}\xspace}

\newcommand{\GroundFacts}{\ensuremath{\mathcal{G}}\xspace}

\newcommand{\Subst}{\ensuremath{\sigma}\xspace}

\newcommand{\Names}{\ensuremath{\calX}\xspace}
\newcommand{\StoreA}{\ensuremath{\calS}\xspace}
\newcommand{\Processes}{\ensuremath{\calP}\xspace}
\newcommand{\ActiveLocks}{\ensuremath{\calL}\xspace}

\newcommand{\pr}{\ensuremath{\mathit{pr}}\xspace}
\newcommand{\Pred}{\ensuremath{\mathit{Pred}}\xspace}

\newcommand{\pout}{\ensuremath{\mathsf{out}}\xspace}
\newcommand{\pin}{\ensuremath{\mathsf{in}}\xspace}

\newcommand{\pevent}{\ensuremath{\mathsf{event}}\xspace}

\newcommand{\msrA}{\mathrel{-\hspace{-2pt}[}}
\newcommand{\msrB}{\mathrel{]\hspace{-4pt}\to}}
\newcommand{\msrewrite}[1]{\msrA#1\msrB}

\newcommand{\traces}{\ensuremath{\mathit{traces}}\xspace}
\newcommand{\tracesmsr}{\ensuremath{\mathit{traces}^\mathit{msr}}\xspace}

\newcommand{\tr}{\ensuremath{\mathit{tr}}\xspace}

\newcommand{\temp}{\ensuremath{\mathit{temp}}\xspace}

\newcommand{\Dom}{\mathfun{dom}}

\newcommand\sem[1]{{\llbracket {#1} \rrbracket}}

\newcommand{\state}{{\sf state}}

\newcommand{\Event}{\ensuremath{\mathit{Event}}\xspace}

\newcommand{\hide}{\ensuremath{\mathit{hide}}\xspace}
\newcommand{\filter}{\ensuremath{\mathit{filter}}\xspace}

\newcommand{\Ass}{\ensuremath{\alpha}\xspace}

\newcommand{\AssIn}{\ensuremath{\alpha_\mathit{inev}}\xspace}

\newcommand{\theactualrule}[1]{\text{Please redefine the command
theactualrule.}}
\newcommand{\underscorethingy}[1]{\text{Please redefine the command
underscorethingy.}}

\DeclareMathOperator\E E

\newcommand{\code}[1]{\normalfont \textsf{#1}\xspace} 

\newcommand{\sign}{\mathfun{sig}}

\newcommand{\Log}{\mathfun{Log}}
\newcommand{\Execute}{\mathfun{Execute}}

\newcommand{\NormalAction}{\mathfun{NormalAct}}
\newcommand{\ExceptionalAction}{\mathfun{ExceptionalAct}}

\newcommand{\verify}{\mathfun{verify}}

\newcommand{\true}{\mathfun{true}}
\newcommand{\commit}{\mathfun{commit}}
\newcommand{\open}{\mathfun{open}}
\newcommand{\verCom}{\mathfun{verCom}}

\newcommand{\oracle}{\mathfun{oracle}}

\newcommand{\CommEv}{\mathfun{Comm}}
\newcommand{\Init}{\mathfun{Init}}
\newcommand{\Final}{\mathfun{Final}}
\newcommand{\HonestS}{\mathfun{HonestS}}

\newcommand{\DishonestS}{\mathfun{DishonestS}}
\newcommand{\DishonestA}{\mathfun{DishonestA}}
\newcommand{\Verified}{\mathfun{Verified}}

\newcommand{\B}{\mathrm{S}}

\newcommand{\Agents}{\mathcal{A}}
\newcommand{\Trusted}{\mathcal{T}}

\newcommand{\verdict}{\mathfun{verdict}}
\newcommand{\apv}{\mathfun{apv}}

\newcommand{\Tamarin}{Tamarin\xspace}
\newcommand{\corrupted}{\mathfun{corrupted}}

\newcommand{\Corrupt}{\mathfun{Corrupt}}
\newcommand{\Control}{\mathfun{Control}}
\newcommand{\Stop}{\mathfun{Stop}}

\newcommand{\Verifiability}[1][i]{\mathname{V}_{\omega_{#1},V_{#1}}}
\newcommand{\Exhaustiveness}{\mathname{XH}}
\newcommand{\Exclusiveness}{\mathname{XC}}
\newcommand{\Uniqueness}[1][i]{\mathname{U}_{\omega_{#1},V_{#1}}}
\newcommand{\UniquenessSingleton}[1][i]{\mathname{U}_{\omega_{#1},V_{#1}=\set{\B}}}
\newcommand{\Minimality}[1][i]{\mathname{M}_{\varphi,V_{#1}}}
\newcommand{\MinimalitySingleton}{\mathname{M}_{R,V_i=\set{\B},\omega_j}}
\newcommand{\MinimalityR}{\mathname{M}_{R,|V_i|\ge 2}}
\newcommand{\Sufficiency}{\mathname{SF}_{\omega_i,\varphi,V_i}}
\NewDocumentCommand{\SufficiencyB}{ O{i} O{\B}}{\mathname{SF}_{\omega_{#1},\varphi,{#2}}}
\newcommand{\SufficiencySingleton}{\mathname{SFS}_{\omega_i,\varphi,V_i=\set{\B}}}
\newcommand{\SufficiencyR}{\mathname{SFR}_{R,\omega_i,\varphi,|V_i|\ge 2}}
\newcommand{\RelationLifting}[1][j]{\mathname{RL}_{R,\omega_i,\omega_{#1},V_i,V_{#1}}}
\newcommand{\RelationReflexive}{\mathname{RS}_{R,V_i=\set{\B}}}
\newcommand{\Completeness}[1][i]{C_{V_{#1}}}
\newcommand{\CompletenessR}{C_{|V_i|\ge 2}}

\newcommand{\eid}{\mathvalue{eid}}
\newcommand{\AssSeq}{\alpha_\mathname{seq}}

\usepackage{pifont} 
\newcommand{\xmark}{\ding{55}}%
 
\newcommand{\fullversionref}{accver-full}
\usepackage{xr} 
\newcommand{\fvref}[1]{%
    \processifversion{conf}{\ref{F-#1} in the full
    version~\cite{\fullversionref}}%
        \processifversion{full}{\ref{#1}}%
        }

\begin{document}

\ifanonymous \else 
\author{%
\IEEEauthorblockN{%
Robert K\"{u}nnemann,  Ilkan Esiyok and Michael Backes} 
\IEEEauthorblockA{CISPA Helmholtz Center for Information Security\\
Saarland Informatics Campus}
}
\fi
\title{Automated Verification of Accountability in Security Protocols}
\begin{full}
\date{\today}
\end{full}
\maketitle

\begin{abstract}
    Accountability is a recent paradigm in security protocol design which aims to 
eliminate traditional trust assumptions on parties and hold 
them accountable for their misbehavior.
It is meant to establish trust in the first place and to
recognize and react if this trust is violated.
In this work, we discuss a protocol-agnostic definition of accountability: 
a protocol provides accountability (w.r.t. some security property)
if it can identify all misbehaving parties, where 
misbehavior is defined as a deviation from the protocol that causes
a security violation.

We provide a mechanized method for the
verification of accountability and demonstrate its use for
verification and attack finding on various examples from the
accountability and causality literature, including Certificate Transparency and
Kroll’s Accountable Algorithms protocol. 
We reach a high degree of automation by expressing accountability in terms of
a set of trace properties and show their soundness and completeness.

\end{abstract}

\section{Introduction}
\label{sec:intro}


The security of many cryptographic protocols relies on trust in some
third parties.
Accountability protocols seek to establish and ensure this trust by
deterrence, often by detecting participants that behave dishonestly.
%
Providing accountability can thus strengthen 
existing properties in case
formerly trusted parties deviate, e.g, the tallying party in an
electronic voting protocol or the PKI in a key-exchange protocol.
Examples of protocols where accountability is relevant include
Certificate Transparency~\cite{rfc6962} to achieve accountability for the
PKI,
OCSP stapling~\cite{rfc6066} to reach accountability for digital 
certificate status requests in PKI 
and Accountable Algorithms~\cite{krollphd}, a proposal for accountable
computations performed by authorities.

We regard
accountability as a meta-property: given some
traditional security property $\varphi$, a protocol that provides
accountability for $\varphi$ permits a specific party or the public to
determine whether a violation of $\varphi$ occurred, and if that is the
case, which party or parties should be held accountable. 
Accountability provides an incentive for `trusted' parties
to remain honest, and allows other parties to react to possible
violations, e.g.,
by revocation or fault recovery. It can also be used to build deterrence
mechanisms.

For a long time, accountability in the security setting lacked
a protocol-independent definition.
Generalized definitions and algorithms have been proposed for
distributed systems~\cite{Haeberlen:2007:PPA:1323293.1294279},
where a complete
view of every component is available. 
In the security setting, however,
the problem of 
identifying dishonest parties is much harder,
as they might deviate from the protocol in an invisible manner.
In unpublished work, Künnemann \etal~\cite{cryptoeprint:2018:127}
approach this problem using causal reasoning. 
Instead of identifying all participants that, perhaps invisibly,
deviate from their specifications;
they identify the parties 
that \emph{cause}
a violation.
Even if a limited part of the communication is available,
cryptographic mechanisms such as digital signatures, commitments and
zero-knowledge proofs
can be used to 
leave traces
that can be causally related to the violation, e.g., a
transmission of a secret.

In this paper, we provide the first mechanized verification technique
for their approach, which was stated in a
custom process calculus in
which several parties can choose individual ways of misbehaving.
First, we propose a variant of their definition in what we call the
\emph{single-adversary} setting. In this setting, a single adversary
controls all dishonest parties. This setting is used in almost all existing
protocol verification tools.
This vastly simplifies the definition of accountability and enables
the use of off-the-shelf protocol verifiers.
Second, we give verification conditions implying accountability for 
a specific \emph{counterfactual relation}. This relation links
\emph{what actually happened} to \emph{what
could have happened}, e.g., if only a subset of parties had
mounted an attack. Causation and, as we will demonstrate,
accountability depend on how this relation is specified.
We use these verification conditions and off-the-shelf tools to automatically
verify Certificate Transparency, OCSP stapling and other protocols for
accountability w.r.t.\ this relation. 
However, more complex scenarios require specifying more fine-grained
counterfactual relations.
We show for the general case: this definition can be decomposed into several
trace properties and the decomposition is
sound and complete.
For our case studies, the verification conditions consist of
7 to 31 such properties; most of them can be verified 
within tenths of seconds.
Third, 
we implement this verification method for an extension of the
applied-$\pi$ calculus to provide a convenient toolchain for the
specification and verification of accountability protocols.
We verify accountability or find attacks on
\begin{inparaenum}[\itshape a\upshape)]
\item several toy examples that highlight the complexity of
    accountability,
\item an abstract modelling of the case where accountability is
    achieved by maintaining a central audit log,
\item several examples from the causality literature, 
\item OCSP-Stapling, 
\item Certificate Transparency and
\item \emph{Accountable Algorithms}, a seminal protocol for 
   accountable computations in a real-world setting.
\end{inparaenum}
\begin{notextended}
We thus list our contributions as follows:
        \begin{itemize}
            \item 
                a new definition of accountability in the
                single-adversary
                setting (which is simpler than the previous
                definition~\cite{cryptoeprint:2018:127} due to the
                single-adversary setting).

            \item a verification technique, based on a 
                sound and complete translation to a set of trace
                properties
                that is compatible with 
                a mature and
                stable toolchain for protocol verification.

            \item a demonstration on several case studies, 
                including
                two
                fixes to 
                Kroll's Accountable Algorithms protocol,
                and a machine-verified proof
                that both fixes entail accountability for both
                participants.
        \end{itemize}
\end{notextended}


\section{Accountability}

To define accountability, we follow the intuition that
a protocol provides accountability if it can always determine
which parties caused a security violation.
In this sense, accountability is a meta-property: we speak of
accountability for violations of $\varphi$.
If no violation occurs, no party causes it;
if no party causes a violation, there is no violation.
Hence accountability w.r.t.\ $\varphi$ implies verifiability of
$\varphi$~\cite{Kusters:2010}.

To reason about failures, i.e., violations of $\varphi$, our formalism
has to allow parties to deviate from the protocol. 
Each party is either \emph{honest}, i.e., it
follows the behaviour prescribed by the
protocol, or it is
\emph{dishonest}, i.e., it may \emph{deviate} from this
behaviour.
Often there is a judging party, which is typically \emph{trusted},
i.e., it is always honest. 
A dishonest party is not necessarily deviating,
it might run a process that is indistinguishable from its protocol behaviour.
Hence it is impossible to detect all dishonest parties.

In the security setting, parties cannot monitor each other
completely. Typically, they only receive messages addressed or redirected to
them. This includes the judging party. Therefore, a deviating party
$A$ can send a secret to another deviating party $B$ and the judging party
may not notice.
Under these circumstances, identifying all deviating
parties is also impossible.

Instead, we focus on identifying all parties \emph{causing} the violation
of the security property $\varphi$.
%
%
The protocols we consider in this work are designed in a way the parties
that are deviating (and thus dishonest) will have to leave evidence in
order to cause security violations.

The definition we provide here is a simplified version of 
an earlier, causality-based definition~\cite{cryptoeprint:2018:127}, 
adapted to a setting where there is only a single adversary
controlling all deviating parties, or equivalently, where the
deviating parties share all the knowledge they obtain.
Both definitions are based on sufficient
causation~\cite{Datta2015a,DBLP:journals/corr/abs-1710-09102}. 
%
The intuition is to capture all parties for which the
\emph{fact that they are deviating at all} is causing the violation.
This may not only be a single party, but also a set of parties. 
If two parties $A$, $B$ collude against a secret sharing scheme
with a threshold of two and expose some secret, the
single party $A$ would not be considered
a cause, but we would say that $\set{A,B}$ have \emph{jointly} caused
the failure. 

Assume a fixed finite set of parties $\Agents=\set{A,B,\ldots}$.
Intuitively, the fact that a party or a set
of parties $\B \subseteq \Agents $ are deviating is a cause for
a violation
iff:
\begin{description}
    \item [SC1.] A violation indeed occurred and $\B$ indeed
        deviated.
    \item [SC2.] If all deviating parties, \emph{except the parties in $\B$}, 
        behaved honestly,
        the same violation would still occur. 
    \item [SC3.] $\B$ is minimal, i.e., SC1 and SC2 hold for no strict subset of $\B$.
\end{description}
The first condition is self-explanatory. The second formalizes that
the misbehaviour of the parties in $\B$ alone is \emph{sufficient} to
disrupt the protocol. 
For the secret sharing scheme scenario above, if the party $A$ deviated alone, this would not cause 
the exposure of the secret, so SC2 would not hold. $\set{A,B}$, however,
would meet all of the conditions.
SC2 reasons about a scenario that is different from the
events that actually took place, which is called
a \emph{counterfactual} in causal reasoning.
At the end of this section, we will discuss different
\emph{counterfactual relations} 
between the actual and the counterfactual course of
events.
The third condition SC3 ensures minimality. It removes parties that
deviated, but whose deviations are not relevant to the coming about of
the violation.

Note that there can be more than one (joint) cause. If the aforementioned
protocol would run in two sessions, one in which $A$ and $B$ colluded
to expose the secret and another one in which $A$ and $C$ colluded to expose 
a different secret, there would be two causes, $\set{A,B}$ and
$\set{A,C}$, each individually satisfying SC1 to SC3, each being an
individual cause for the failure.
This nicely separates joint causes ($A$ and $B$ working together;
$A$ and $C$ working together) from independent causes (the first
collusion and the second collusion).

Intuitively, an
accountability mechanism provides accountability for $\varphi$ if
it can always point out all $\B \subseteq \Agents$ causing a violation
of $\varphi$, or $\emptyset$ if there is no violation.

\subsubsection*{Need for trusted parties} 

Accountability protocols aim at eliminating trust assumptions, yet
they often require a trusted judging party. 
But, if the judging party operates offline and their judgement
can be computed with public knowledge, then any member of the public
can replicate the judgement and validate the result. The judging party is thus
not a third party, but the observer itself, who we assume can
be trusted to follow the protocol, i.e., not cheat herself.

\subsubsection*{Individually deviating parties vs.\ single adversary} 

Protocol verification
tools usually consider a single adversary representing both an
intruder in the network and untrusted parties. Typically, the attacker
sends a \emph{corruption} message to these
parties, they accordingly transmit their secrets to the 
adversary. 
As messages are transmitted over an insecure network, the 
adversary can then impersonate the parties with their secrets. 
This is a sound over-approximation in case of security properties
like secrecy and authentication, as the adversaries are stronger if they
cooperate. It also simplifies verification: a single adversary
can be modelled in terms of non-deterministic state transitions for
message reception (see, e.g.,~\cite[Theorem~1]{AF-popl01}).
Many protocol verification tools operate in this
way~\cite{SMCB-csf12,blanchetcsfw01,
Cr2008Scyther,
MaudeNPA09}.
We therefore 
focus on the single-adversary setting, in
order to exploit existing methods for protocol verification
and encourage the development of accountability
protocols. 
This is the main distinction between our simplified definition of
accountability and the earlier definition~\cite{cryptoeprint:2018:127}.
Nevertheless, the philosophical difference between
these two settings remains relevant for our analysis, as it is based
on causal reasoning about deviating parties. In
Section~\ref{sec:related}, we discuss the difference between
these two settings and elaborate on our current understanding of them.

\subsubsection*{Protocol accountability}

Let a protocol $P$ be defined in some calculus, e.g., the applied-$\pi$
calculus, and assume that we can define the set of traces induced by
the protocol, denoted $\traces(P)$.
We also assume a \emph{counterfactual relation} $r$ between
an actual trace and a counterfactual trace. Depending on the
protocol, different counterfactual relations can capture the desired
granularity of a judgement.
We will describe this relation and give examples in the following section.

Given a trace $t$, a security property $\varphi$ can be evaluated, e.g.,
$t \vDash\varphi$ iff $\varphi$ is true for $t$,
which we will sometimes abbreviate to $\varphi(t)$.
We assume that any trace $t\in\traces(P)$ determines
the set of dishonest agents,
i.e., those who received the corruption message from the adversary
and define $\corrupted(t)\subseteq \Agents$ to be this set.

We can now define the a~posteriori verdict (apv), which specifies all
subsets of $\Agents$ that are sufficient to cause $\neg\varphi$.  
The task of an accountability protocol is to always compute the $apv$,
but without having full knowledge of what actually happened, i.e.,
$t$.

\begin{definition}[a~posteriori verdict]\label{def:a-posteriori-verdict} 
    Given 
    an accountability protocol $P$,
    a property $\varphi$,
    a trace $t$ and
    a relation $r$,
    the a~posteriori verdict,
    which formalizes the set of parties causing $\neg\varphi$,
    is defined as:
    \begin{multline*}
        \apv_{P,\varphi,r}(t) := 
        \{\B\mid\text{$t \vDash \neg \varphi$ and 
          $\B$  minimal s.t.}\\
    \exists t' : 
    r(t,t') \wedge \corrupted(t') = \B 
    \wedge t' \vDash \neg\varphi\}.
    \end{multline*}
    Relation $r$ is
    reflexive, transitive and
    $r(t,t')$ implies 
    $\corrupted(P,t') \subseteq \corrupted(P,t)$.
\end{definition}

Each set of parties $\B\in\apv_{P,\varphi,r}(t)$ is a sufficient cause
for an eventual violation, in the sense outlined at the beginning of
this section.
The condition $t \vDash \neg \varphi$ ensures that indeed a violation
took place. If the parties in $\B$ did not deviate in $t$,
then $\B$ would not be minimal, hence
SC1 holds. SC3 holds by the minimality condition. The remaining
conditions capture SC2: $t'$ is a \emph{counterfactual trace} w.r.t.\
$t$, 
a trace contrary to what actually happened.\footnote{Technically, the
set of counterfactual traces includes $t'=t$, as $r$ is reflexive, and
thus not every instance of $t'$ is, strictly speaking,
a counterfactual. For brevity, we prefer to nevertheless call $t'$ a counterfactual, despite this imprecision.}
In
$t'$,
only the parties in $\B$ may deviate, which should 
suffice to derive a violation.
We define the relation $r$ to constrain the set of counterfactual
traces.
At the very least, 
the condition that $\corrupted(t') \subseteq \corrupted(t)$ should hold
to guarantee that for any violating trace, a minimal $S$ is
defined.\footnote{Consider $r$ total, $t$ with $t\vDash \neg \varphi$
and $\corrupted(t)=\set{A}$ and $t'$ with $t'\vDash \neg \varphi$ and
$\corrupted(t')=\set{B}$ as a counterexample. Neither $\set{A}$ nor
$\set{B}$ would be minimal.}
We will discuss a few variants of $r$ in the following section.

For a given trace,
the apv outputs a set of sufficient causes, i.e., a set of sets of
agents. We call this output as the \emph{verdict}.
We also remark that $\apv_{P,\varphi,r}(t)=\emptyset$ iff
$t\vDash\varphi$, i.e., an empty verdict means the absence of
a violation --- there can be no cause for an event that did not happen.

To abstract from the mechanism by which an accountability protocol
announces the purported set of causes for a violation 
--- this could range from a designated party computing them to a public
ledger that allows the public to compute it ---
we introduce a \emph{verdict function}
from $\traces(P) \to 2^{2^\Agents}$. 
An accountability protocol is thus described by $P$ and the verdict
function.
We can now state accountability: a verdict function provides
accountability if it always computes the apv.

\begin{definition}[accountability]\label{def:accountability} 
    A verdict function
    $\verdict: \traces(P) \to 2^{2^\Agents}$
    provides 
    a protocol $P$ 
    with
    accountability for a property $\varphi$
    (w.r.t.\ $r$)
    if, for any trace $t\in\traces(P)$
    \[ \verdict(t) = \apv_{P,\varphi,r}(t). \]
\end{definition}

\begin{example}\label{ex:mon-fst}
Assume a protocol in which a centralized monitor controls access to some resource,
e.g., confidential patient records. Rather than preventing exceptional access
patterns, e.g., requesting the file of another doctor's patient,
the requests are allowed, but logged by the monitor. Accountability tends to be preferable over prevention in
such cases,
e.g., in case of emergencies.

The set of agents comprises doctors $D_1, D_2$ and the centralized monitor.
The centralized monitor is trusted and effectuates requests only if
they are digitally signed. The security property $\varphi$ is true if
no exceptional request was effectuated. 
Per protocol, $D_1$ and $D_2$ never send exceptional requests, however,
if in trace $t$, the central adversary corrupts them and learns their signing
keys,
he can act on their behalf and
sign exceptional requests. The set of dishonest
parties contains those corrupted by the adversary, 
    $\corrupted(t)=\set{D_1,D_2}$.
A verdict $\set{\set{D_1},\set{D_2}}$ indicates that both $D_1$
and $D_2$ deviated in a way that caused a violation, i.e., 
an exceptional request was effectuated. 
If a third doctor $D_3$ was also in the protocol and
it was corrupted $D_3\in\corrupted(t)$, but never signed any request, 
then it would not be involved in the apv
    $\apv_{P,\varphi,r}(t)=\set{\set{D_1},\set{D_2}}$.
The apv for this protocol can be computed by only taking the log into account: a verdict
function $\verdict$, that operates on the monitor's log, can be defined to compose singleton verdicts 
for each party that signed an exceptional request.
This verdict function is easy to implement, yet it provides the protocol
with accountability for $\varphi$ because it computes
    $\apv_{P,\varphi,r}$ correctly.
\end{example}

\subsubsection*{Counterfactual relation}\label{sec:relations}

The relation $r$ in the definition of the apv defines which
counterfactual scenarios are deemed relevant in SC2.
While there is an agreement in the causality literature that $t'$
cannot be chosen arbitrarily, there is no agreement on \emph{how}
they should relate.
We slightly change the previous monitoring example and
discuss why the relation is important.
\begin{example}\label{ex:mon-snd}
    Assume that the monitor supports a second mechanism
    to handle requests. Here, a doctor $D_1$ can also sign his exceptional request and
    ask his chief of medicine
    $C$
    to approve it.
    Assume in trace $t$, both $D_1$ and $C$ collude and use this
    mechanism to effectuate an exceptional request, violating $\varphi$.
    Intuitively, one would expect the apv, relying on logs, to give the verdict
    $\set{\set{D_1,C}}$. However, $D_1$ could have used the \emph{first} mechanism for this request. Hence, there is
    a counterfactual trace $t'$ where only $D_1$ is dishonest.
    If $r(t,t')$, then the more intuitive verdict $\set{\set{D_1,C}}$ is not minimal,
    but $\apv_{P,\varphi,r}(t)=\set{\set{D_1}}$ is minimal, shifting the blame to $D_1$ alone.
    The intuitive response would be: `But that is not what happened!',
    which is precisely what $r$ needs to capture.
\end{example}

We discuss 
three
approaches for relating factual and
counterfactual traces: 

\begin{itemize}
    \item\emph{by control-flow:} $r_c(t,t')$ iff $t$ and $t'$
        have similar control-flow.
        Several works in the causality literature relate traces by their
        control-flow~\cite{Datta2015a,cryptoeprint:2018:127,causality-faulttrees},
        requiring counterfactuals to retain, to varying degree, the same control-flow.
        See~\cite{causality-control-flow} for a detailed discussion
        about control flow in Pearl's influential causation framework.
        %
        %
        \begin{full}
        A key challenge here is that parties that are corrupted in
        $t$, but now honest in $t'$ need to be able to choose
        a different control-flow. 
        %
        %
        In many protocol verification tools, the 
        adversary is not represented as a process with control-flow,
        but as an environment that non-deterministically chooses
        message input constrained by a deduction relation.
        Thus, there is no adversarial control-flow, which means that in
        the single-adversary setting, we can only talk about the
        control-flow of honest parties.
        This is natural: we can only observe deviating parties, but
        not know which program they run.
        \end{full}
        For simplicity, the notion we present
        (cf.\ Section~\ref{sec:ctl-flow})
        captures only the
        control-flow of trusted parties, i.e., parties guaranteed
        to be never controlled by the adversary.
        In case of the example, the control-flow of the trusted
        monitor would distinguish these two mechanisms.

    \item \emph{by kind of violation:}
        $r_k(t,t') $ iff $t$ and $t'$ describe the same kind of
            violations.
       This approach is, e.g.,  
       used in criminal law 
       to solve causal problems where the classical `what-if` test
       fails, e.g., 
       a person was poisoned, but shot before the poison took effect.
       %
       The classical test for causality (sine qua non) gives
       unsatisfying results (without the shot, the person would still
       have died), unless one describes the causal event in
       more detail, i.e., by distinguishing death from shooting from death
       from poisoning (~\cite[p.~188]{dressler1995understanding}; see
also~\cite[p.~46]{Mackie1980-MACTCO-17}). 
        For security protocols, the instance of the violation
        could be characterized by the session, the
        participating parties or other variables that are free in
        in the security property $\varphi$ (see Lowe~\cite{lowe1997hierarchy}).
       This relation is informal and depends on intuition, 
       so it is not used in our analysis.

    \item \emph{weakest relation according to
        Def.~\ref{def:a-posteriori-verdict}}: $r_w(t,t') \iff
        \corrupted(P,t')\subseteq \corrupted(P,t)$.
        This relation is conceptually simple and suitable for many 
        protocols in which collusion is not an issue, i.\,e.,
        verdicts contain only singleton sets.
        Outside this class, it
        may give unintuitive a~posteriori verdicts in cases where $t$ 
        requires
        collusion, but one of the colluding parties could mount
        a possibly very different attack by themselves.
        %
\end{itemize}

In Section~\ref{sec:cond-coarse}, we provide verification conditions for
$r_w$. In Section~\ref{sec:ver}, we provide more general 
verification conditions that
apply to arbitrary relations, as some scenarios require more 
fine-grained analysis, and later we mechanize it for $r_c$.

\begin{full}
Besides these relations, there are also approaches in distributed systems that 
require
counterfactuals to reproduce the same observations for some amount of
time (e.g., for as long as possible, or as long as it complies to come
specification) and then switch to their specified
behavior~\cite{gossler:hal-00924048,causality-reactive}.
\end{full}

\section{Verification conditions for $r_w$}\label{sec:cond-coarse}

In this section, 
we define a set of verification conditions
parametric in a security property $\varphi$ and a verdict function.
If these conditions are met, they
provide a protocol with accountability for
$\varphi$ w.r.t.\ the weakest condition on counterfactuals
$r_w(t,t') \defeq
\corrupted(t') \subseteq \corrupted(t)$.
Each of these conditions is a trace property and can thus be verified
by off-the-shelf protocol verification tools.
In our case studies, we will use these verification conditions to verify
accountability properties for the Certificate Transparency protocol.

The main idea:  we assume the verdict function is described as a case
distinction over a set of trace properties $\omega_1$ to $\omega_n$.
Any of these \emph{observations}
$\omega_i$ is then 
assigned a verdict $V_i$.

\begin{definition}[verification conditions]\label{def:verification-conditions-coarse} 
    Let $\verdict$ be a verdict function of form:
    \[
        \verdict(t) = \begin{cases}
            V_1 & \text{if $\omega_1(t)$}\\
            \vdots &  \\
            V_n & \text{if $\omega_n(t)$}\\
        \end{cases}
    \]
    and $\varphi$ a predicate on traces.
    We define the verification condition
    $\gamma_{\varphi,\verdict}$ 
    as the conjunction of the formulae in Table~\ref{tab:coarse-cond}.

    \begin{table}[t]
      \centering
      
        \caption{Verification Conditions $for \: rw$.}\label{tab:coarse-cond}
        \begin{tabular}{lr}
            conditions & formulae\\
            \toprule
            \emph{Exhaustiveness ($\Exhaustiveness$):} & $\forall t. \quad \omega_1(t)\lor \cdots \lor \omega_n(t)$ \\
             \midrule 
            \emph{Exclusiveness ($\Exclusiveness$):} & $\forall t, i, j. \quad i \neq j \implies \neg(\omega_i(t)\land  \omega_j(t))$ \\
             \midrule 
            \emph{Sufficiency of each $\omega_i$} & $\forall \B\in V_i. \quad \exists t.$ \\
            \emph{s.t.\ $V_i\neq\emptyset$ ($\SufficiencyB$):} & $\neg\varphi(t) \land \corrupted(t) = \B $ \\
             \midrule 
            \emph{Verifiability of each $\omega_i$} & $\forall t. \quad \omega_i(t) \implies$ \\
            \emph({$\Verifiability$}): & $ (V_i=\emptyset \iff \varphi(t))  $ \\
             \midrule 
            \emph{Minimality of each $V_i$} & $\forall \B \in V_i \ \forall \B' \subsetneq \B  $ \\
            \emph({$\Minimality$}): & $ \nexists t.\ \neg\varphi(t) \land \corrupted(t) = \B' $  \\
             \midrule 
            \emph{Uniqueness of each $V_i$} & $\forall t. \quad \omega_i(t) \implies$ \\
            \emph({$\Uniqueness$}): & $ \bigcup_{\B\in V_i} \B \subseteq \corrupted(t) $ \\
            \midrule 
            \emph{Completeness of each $V_i$}
            &
            $\forall \B\subseteq \bigcup_{\B'\in V_i} \B' \ \forall j.\ V_j=\set{\B} \implies$ 
            \\
            \emph{($\Completeness$):}
            &
            $\B\in V_i$.
            \\
            \bottomrule
        \end{tabular}
        \vskip 0.1in
        ($t$ is quantified over $\traces(P)$)
    \end{table}

\end{definition}

We briefly go over these conditions.
The case distinction needs to be exhaustive
($\Exhaustiveness$) and exclusive ($\Exclusiveness$), because 
verdict functions are total.
For any observation $\omega_i$ that leads to a non-empty verdict, any
set of parties $\B$ in this verdict needs to be able to produce
a violating trace on their own ($\SufficiencyB$).
However,
removing any element from $\B$ should make it impossible to produce
a violation ($\Minimality$), due to the minimality requirement
in Def.~\ref{def:a-posteriori-verdict}.
If an observation leads to the empty verdict, it needs to imply the
security property $\varphi$, because accountability implies
verifiability ($\Verifiability$).
Whenever an observation $\omega_i$ is made, all parties
that appear in the ensuing verdict have necessarily been corrupted
($\Uniqueness{}$). This ensures uniqueness; if there was a second
sufficient and minimal verdict, part of the verdict would correspond
to a trace that corrupts parties that do not appear in the verdict
(details in the proof of completeness,
Appendix~\fvref{sec:proof-cond-coarse}).
Finally, if there is a singleton verdict 
(e.g., $V_j = \set{\set{B,C}}$)
containing only parties that appear in another composite verdict
(e.g., $V_i = \set{\set{A,B},\set{A,C}}$)
then traces that give the former are related to traces that give the
latter (where, at least, $A$, $B$ and $C$ were dishonest).
Hence the singleton verdict needs to be included.
($\Completeness{}$).

We show these conditions
sound and complete
in Appendix~\fvref{sec:proof-cond-coarse}.
Practically, this means that any counter-example to any lemma
generated from these conditions demonstrates an attack against
accountability.
\begin{example}\label{ex:mon-coarse}
    Consider the centralized monitor from Example~\ref{ex:mon-fst},
    and, for simplicity, assume there is only one doctor $D$. The verdict
    function gives $V_1=\set{\set{D}}$ if it logged an action signed by
    $D$, if this action was effectuated and if it was exceptional.
    Otherwise, it gives $V_2=\emptyset$.
    To show that his verdict function provides accountability
    for $\varphi \defeq \text{no exceptional action was effectuated}$,
    one would show:
    \begin{itemize}
        \item the case distinction exhaustive and exclusive,
        \item that the attacker can effectuate an exceptional action if $D$'s signing key is known ($\SufficiencyB$),
        \item that the `otherwise' condition (no exceptional action
            was effectuated and signed by $D$) implies that no
            exceptional
            action was effectuated by anyone ($\Verifiability$),
        \item that no exceptional action can be effectuated without knowledge
            of $D$'s signing key ($\Minimality$), and
        \item that $D$'s signature on an exceptional action that was
            effectuated can only be obtained by corrupting $D$
            ($\Uniqueness$).
        \item Completeness ($\Completeness$): $V_1$ is the
            only non-empty verdict.
    \end{itemize}
\end{example}

\section{Verification conditions for arbitrary $r$}\label{sec:ver}

As outlined in Example~\ref{ex:mon-snd},
there are scenarios where a more fine-grained analysis is necessary.
These scenarios are characterized by violations that can be provoked
either by a set of colluding parties or by a subset thereof, using
a different mechanism.
Hence we provide a different and more elaborate set of verification
criteria (see Table~\ref{tab:ver-conditions}). They fall into two categories: the first consists of
trace properties that again can be verified using off-the-shelf
protocol verifiers. The second relates the case distinction used to
define the verdict to the relation: in general, all traces that
fall into the same case should be related. The verification of the
second kind of
conditions depends on the relation chosen and can be conducted by
hand. In a later section, we mechanize the verification of these
conditions for the relation $r_c$ specifically. 
Hence our method is fully automated for this relation.

These verification criteria are sound and complete.
(see Appendix~\fvref{sec:proof-cond-cases}).
This means that the conjunction of all verification criteria is
logically equivalent to accountability and thus contradicts
Datta et.\,al.'s view that `accountability depends on actual
causation and it is not a trace
property'~\cite{Datta2015a}. 

\begin{definition}[verification conditions]\label{def:ver-conditions} 
    Let $\verdict$ be a verdict function of 
    \begin{full}
    form:
    \[
        \verdict(t) = \begin{cases}
            V_1 & \text{if $\omega_1(t)$}\\
            \vdots &  \\
            V_n & \text{if $\omega_n(t)$}\\
        \end{cases}
    \]
    \end{full}
    \begin{conf}
        the form from
        Definition~\ref{def:verification-conditions-coarse}
    \end{conf}
    and $\varphi$ be a predicate on traces.
    We define the verification condition
$\nu_{\varphi,\verdict}$ 
    as the conjunction of the formulae in Table~\ref{tab:ver-conditions},
    where $t$ and $t_i$ range over $\traces(P)$. 
\end{definition}

Again, we assume the verdict to be expressed as a case distinction.
This 
case distinction must be
sufficiently fine-grained to capture all relevant classes of
counterfactuals, e.g., all ways the violation could come about in
terms of $r$. 
We distinguish between the empty verdict 
(meaning no violation took place),
singleton verdicts ($\set{\B}$, where $\B$ itself
is a set of parties jointly causing a violation) 
and composite verdicts (consisting of two or more elements, e.g.,
$\set{\set{A,B},\set{A,C}}$ if $A$, $B$ and $C$ deviated, but $A$ 
could have caused the violation
either jointly with $B$ or with $C$).
The main idea is that the correctness of composite verdicts, e.g.,
$\set{\set{A,B},\set{C}}$,
follows from the correctness of the singleton verdicts they are
composed from, e.g., $\set{\set{A,B}}$ and $\set{\set{C}}$, as long as all
traces that provoke the composite verdict relate to the singleton
verdict.
        \begin{table}[t!]
      \centering
      
        \caption{Verification conditions $for \: arbitrary \: r$.}\label{tab:ver-conditions}
        \begin{tabular}{lr}
            conditions & formulae\\
            \toprule
            \emph{Exhaustiveness ($\Exhaustiveness$):} & $\forall t. \quad \omega_1(t)\lor \cdots \lor \omega_n(t)$ \\
             \midrule 
            \emph{Exclusiveness ($\Exclusiveness$):} & $\forall t, i, j. \ i \neq j \implies \neg(\omega_i(t)\land  \omega_j(t))$ \\
             \midrule 
            \emph{Sufficiency for each $\omega_i$}, & $\exists t. \quad \omega_i(t) \land \neg\varphi(t) $ \\
            \emph{with singleton $V_i=\set{\B}$} & $ \land \corrupted(t) = \B $ \\
            \emph({$\SufficiencySingleton$}): & \\
             \midrule
            \emph{Sufficiency for each $\omega_i$}, & $\forall \B. \quad \B\in V_i \implies $ \\
            \emph{with $|V_i|\ge 2$} & $ \exists j. \quad R_{i,j} \land V_j = \set{\B} $ \\
            \emph({$\SufficiencyR$}): &  \\
             \midrule
            \emph{Verifiability for each $\omega_i$}, & $\forall t. \quad \omega_i(t) \implies $ \\
            \emph({$\Verifiability$}): & $ (V_i=\emptyset \iff \varphi(t)) $ \\
             \midrule
            \emph{Minimality, for} & $\forall \B'. \quad \B' \subsetneq \B \implies $ \\
            \emph{singleton $V_i=\set{\B}$} & $ \nexists t. \quad \omega_i(t) \land \corrupted(t) = \B' $ \\
            \emph({$\Minimality$}): & \\
             \midrule
            \emph{Minimality, for} & $ \nexists \B, \B' \in V_i \quad \B' \subsetneq \B$  \\
            \emph{composite $V_i$, $|V_i|\ge 2$} &  \\
            \emph({$\MinimalityR$}): & \\
             \midrule
            \emph{Uniqueness of each $V_i$} & $ \forall t. \quad \omega_i(t) \implies \B \subseteq \corrupted(t) $ \\
            \emph{with $V_i=\set{\B}$} & \\
            \emph({$\UniquenessSingleton$}): &  \\
             \midrule 
            \emph{Completeness of $|V_i|\ge 2$} & $ \forall j,S. \quad R_{i,j} \land V_j=\set{\B} \implies  $ \\
            \emph({$\CompletenessR$}): & $ \exists \B'. \quad \B'\in V_i \land \B \subseteq \B' $ \\
            \midrule
            \emph{Relation is lifting of r} & if $ V_i,V_j \neq \emptyset$ then \\
            \emph({$\RelationLifting$}): & $\forall t,t'. \quad \omega_i(t) \land \omega_j(t')$ \\
            & $ \land R_{i,j} \iff r(t,t') $ \\
             \midrule
            \emph{Relation is reflexive and} & $\forall i,j. \quad V_i$ is singleton \\
            \emph{terminating on singleton} & $ \land R_{i,i} \land ( R_{i,j} \implies i=j) $ \\
            \emph({$\RelationReflexive$}): &  \\
             \bottomrule   
            
        \end{tabular}
        \vskip 0.5em
        ($t$ is quantified over $\traces(P)$)
    \end{table}

We assume the cases to be connected along these lines;
all cases resulting in an empty verdict have to guarantee $\varphi$ to
hold ($\Verifiability$).
All cases resulting in singleton verdicts have to imply that 
\begin{inparaenum}[(a)]
\item a violation took place ($\Verifiability$),
\item that the parties in the verdict alone can provoke this
    violation ($\SufficiencySingleton$),
\item these parties need to be corrupted whenever this case
    matches ($\Minimality$) and
\item that the verdict is unique ($\UniquenessSingleton$).
\end{inparaenum}

Composite verdicts need to relate to singleton verdicts by means of
a lifting $R$ of the relation $r$ ($\RelationLifting$). For each part
of a composite verdict, $R$ points to the singleton verdicts for the
same parties. Therefore, it needs to be terminating on singleton cases
($\RelationReflexive$).
\begin{full}%
\footnote{%
One can visualise the relation between these cases as follows. As $r$
is transitive, we can define equivalence classes for all mutually
related $t$ and $t'$ ($r(t,t')$ and $r(t',t)$). If $t_A$ and $t_B$ are
in different equivalence classes $A$ and $B$, and if $r(t,t')$, then
by transitivity of $r$, all members of $A$ relate to all members of
$B$. 
Hence we can think of these equivalence classes as nodes in a directed
graph with an edges between related equivalence classes. 
This graph is acyclic, as otherwise, all equivalence classes in this
cycle would be in the same equivalence class by transitivity. It is
thus a forest.
By transitivity, each node in each tree in this forest can be reached
in one step from the root. All leaf nodes correspond to singleton
verdicts, and all root nodes that are not leaf nodes themselves are
composite verdicts.}%
\end{full}
As $R$ is a lifting of $r$ and
reflexive (i.e.,  
all traces in the same case 
are related to each other, $\RelationReflexive$),  
completeness, sufficiency
and minimality
carry over from singleton cases to composite cases, as long as
all parts of of the composite verdict are covered by a singleton
verdict ($\SufficiencyR$), and, vice versa, all related singleton verdicts are
contained in the composite verdicts ($\CompletenessR$). 
In this way, we can avoid the requirement that the composite
verdict itself needs to define the minimal set of parties needed to provoke a violation, 
which would not be the case in most of
our case studies.\footnote{Compare with the minimality requirement in
Table~\ref{tab:coarse-cond} in Def.~\ref{def:verification-conditions-coarse}.}
Instead, we only need
a simple syntactic check ($\MinimalityR$) 
to ensure that
the parts of a composite verdict are not contradictory with regards to the
minimality of the apv.
Consider, e.g., the composite verdict $\set{\set{A,B},A}$.

In summary, the key to these verification conditions is to express
equivalence classes w.r.t.\ $r$ and  relations between them
in the case distinction describing the verdict function.

\begin{example}\label{ex:mon-cases}
    Consider the extended centralized monitor
    (Example~\ref{ex:mon-snd}) and assume that the logged signature
    distinguishes which mechanism was used to effectuate an action and
    that only one action can be effectuated.
    Assume further a verdict function that 
    \begin{inparaenum}[(a)]
    \item outputs $\set{\set{D}}$ if
        the monitor logged and effectuated an exceptional action signed by
        $D$,
    \item  outputs $\set{\set{C,D}}$ if it logged and effectuated an
        exceptional action signed by $D$ and $C$, and 
    \item  
        outputs $\emptyset$
        otherwise.
    \end{inparaenum}
    Minimality can only hold if case (a) and case (b) are not in
    the relation. 
    Hence $\RelationLifting$ requires any trace falling into
    case (a) to be unrelated from any trace falling in 
    case (b). $\RelationReflexive$ requires all traces leading to
    the observation in case (a) to be related, and likewise for 
    case (b).
    If we consider, e.g., the monitor's control flow $r_c$,
    this can be shown automatically (cf.\ Section~\ref{sec:ctl-flow}). 
    If we consider $r_k$, it is essentially an axiom.
\end{example}

\section{Calculus}
\label{sec:calculus}

Before we present our case studies, we will elaborate on the protocol
calculus in which they are stated.
The calculus we used 
is an extension of the well-known 
applied-$\pi$ calculus~\cite{AF-popl01}.
In addition 
to
the usual operators for concurrency, replication,
communication, and name creation, 
from applied-$\pi$,
this calculus (called
SAPiC~\cite{KK-jcs16,BaDrKr-2016-liveness}\footnote{
Our results apply to both the original version of SAPiC and the extension with reliable channels.
})
supports 
constructs for accessing and updating
an explicit global state,
which is useful for
accountability protocols that rely on trusted third parties retaining
some state or a public ledger.
Readers familiar with the applied-$\pi$ calculus can jump straight to
Section~\ref{sec:calculus-acc}, where the modelling of corruption is
explained. The constructs for state manipulation are marked in
Figure~\ref{fig:syntax}.

\begin{full}
SAPiC was
extended with constructs 
for non-deterministic choice and communication on a reliable
channel to accommodate the 
analysis of fair-exchange protocols~\cite{BaDrKr-2016-liveness}.
We conjecture this extension to be useful for accountability protocols,
as this is necessary to verify \emph{timeliness} and \emph{fairness}. In an
accountability context, holding parties accountable for timeliness can
assure that any party delaying the computation indefinitely can be 
identified, i.e., if $A$ fails to provide $B$ with data that $A$ needs
to answer $C$, can $B$ prove to $C$ that it is $A$'s fault. 
We conjecture fairness to occur as a necessary condition to other
accountability properties: accountability can be reached by each party
attesting that the other party provides certain data, and identifying
maliciously deviating parties by  missing or inconsistent information.
If $A$ and $B$ collect information on each
other, the first to receive the information proving they behaved
correctly could be in a situation where it might (falsely) blame the
other. A fair sub-protocol can prevent this situation.
%
\end{full}

We will now introduce the syntax and informally explain the semantics of
the calculus. For the formal semantics, please refer to
Appendix~\ref{app:semantics}.



\subsubsection*{Terms and equational theories} Messages are modelled
as abstract terms.  We define an order-sorted term algebra with the sort
\msgsort and two incomparable subsorts \pubsort and \freshsort for
two countably infinite sets of 
public names ($\PN$) 
and fresh names ($\FN$).
Furthermore, we assume a countably infinite set of variables
for each sort $s$, $\Vars_s$. Let \Vars be the union of the set of
variables for all sorts.  We write $u:s$ when the name or variable $u$
is of sort $s$. Let \Sign be a signature, i.e., a set of function
symbols, each with an arity. We write $f/n$ when function symbol $f$
is of arity $n$.  
There is a subset $\Sign_\mathit{priv}\subseteq \Sign$
of \emph{private} function symbols which
cannot be applied by the adversary.
Let $\Terms$ be the set of well-sorted
terms built over \Sign, \PN, \FN and \Vars,
and $\Mess$ be the subset containing only ground terms, i.e., terms
without variables.

Equality is defined by means of an equational theory $\ET$, i.e.,
a finite set of equations between terms inducing a binary relation
$=_\ET$
that is closed under application of function symbols, bijective
renaming of names and substitution of variables by terms of the same
sort.


\begin{example}
    We model digital signatures using symbols
    $\set{\sign,\verify,\pk,\sk,\true}\subset\Sign$ with
    $\sk\in\Sign_\mathit{priv}$,
    and equation
    \[
    \verify(\sign(m, \sk(i)), m, \pk(\sk(i)))=\true.\]
\end{example}

For the remainder of the article, we will assume the signature $\Sign$ and
equational theory $\ET$ to
contain symbols and equations for pairing and projection 
$\{ \pair{.}{.}, \fst, \snd \} \subseteq \Sign$ and equations
$\fst(\pair x y) = x$ and $\snd(\pair x y) = y$ are in \ET.
We use $\langle x_1, x_2, \ldots, x_n\rangle$ as
a shortcut for $\langle x_1,\langle x_2,\langle \ldots,\langle x_{n-1}, x_n\rangle\ldots\rangle$.
We suppose that functions between terms are interpreted modulo $E$, i.e., if $x =_E y$ then $f(x) = f(y)$. 

\subsubsection*{Facts} 
We also assume an unsorted signature \FSign,
disjoint from $\Sign$. The set of \emph{facts} is defined as 
$\calF
:= \{ F(t_1,\ldots,t_k) \mid t_i\in\Terms_\Sign, F\in\FSign \mbox{ of
  arity }k \}$ and
used to annotate protocol steps. 

\subsubsection*{Sets, sequences, and multisets} 
We write $\setN_n$ for the set $\set{1,\ldots, n}$. Given a set $S$, we denote the set of finite sequences of elements from $S$, by $S^*$.
Set membership modulo $E$ is denoted by $\in_E$ and defined as $e \in_E S$ iff $\exists e' \in S.\ e' =_E e$. $\subset_\ET$, $\cup_E$, and $=_\ET$ are defined for sets in a similar way. Application of substitution is lifted to sets, sequences and multisets as expected.  By abuse of notation we sometimes interpret sequences as sets or multisets; the applied operators should make the implicit cast clear.

\subsection{Syntax and informal semantics}
\begin{figure}[th]
    \begin{grammar}
            <$P$,$Q$> ::= 
        0\tikzmark{one}
      \alt $P \mid Q$   
      \alt $!$ $P$
      \alt $\nu a\colon\freshsort$; $P$
      \alt out($m,n$); $P$
      \alt in($m,n$); $P$  
      \alt if $\Pred$ then $P$ [else $Q$]   \tikzmark{three}
      \alt event $F$; $P$ \tikzmark{two}
      \alt $P + Q$   \quad\textit{(non-deterministic choice)}
        \tikzmark{a}
      \alt insert $m$,$n$; $P$ \quad\textit{(inserts $n$ at cell $m$)}
      \alt delete $m$; $P$  \quad\textit{(deletes content of $m$)} \tikzmark{c}
      \alt lookup $m$ as $x$ in $P$ [else $Q$]
      \alt lock $m$; $P$
        \alt unlock $m$; $P$ \tikzmark{b}
    \end{grammar}
    \caption{Syntax ($a\in\FN$, $x\in\Vars$, $m,n\in\Terms$
      $\Pred\in\calP$,
      $F\in\calF$).
}\label{fig:syntax}
\begin{tikzpicture}[overlay, remember picture]
    \draw [decoration={brace,amplitude=0.5em},decorate,thin,black, xshift=1cm]
    let \p1=(pic cs:one), \p2=(pic cs:two), \p3=(pic cs:three) in
    ({max(\x1,\x3)}, {\y1+0.8em}) -- node[right=1.6em, label distance=0.5cm,text depth=-1ex,rotate=-90, anchor=south] {classical applied-$\pi$} ({max(\x1,\x3)}, {\y2});
    \draw [decoration={brace,amplitude=0.5em},decorate,thin,black, xshift=1cm]
    let \p1=(pic cs:a), \p2=(pic cs:b), \p3=(pic cs:c) in
    ({max(\x1,\x3)}, {\y1+0.8em}) -- node[right=1.6em, label distance=0.5cm,text depth=-1ex,rotate=-90, anchor=south] {SAPiC additions} ({max(\x1,\x3)}, {\y2});
\end{tikzpicture}

\end{figure}
$0$ denotes the terminal process. $P \mid Q$ is the parallel execution
of processes $P$ and $Q$ and $! P$ the replication of $P$ allowing an
unbounded number of sessions in protocol executions.  
$P+Q$ denotes
\emph{external} non-deterministic choice, i.e.,
if $P$ or $Q$ can reduce to a process $P'$ or $Q'$, $P+Q$ may reduce
to either.
%
%
The construct $\nu a; P$ binds the name $a\in\FN$ in $P$ and models
the generation of a fresh, random value.
The processes out($m,n$); $P$ and in($m,n$); $P$ represent the output,
respectively input, of message $n$ on channel $m$. 
As opposed to the applied pi calculus~\cite{AF-popl01},  
SAPiC's input construct performs pattern matching instead of variable binding. 
%
If the channel is left out, the public channel $c$ is assumed, which
is the case in the majority of our examples.
%
The process \code{if \Pred then $P$ else $Q$} will execute $P$ or $Q$, depending on whether $\Pred$ holds.  For example, if $\Pred=\mathit{equal}(m,n)$, and $\phi_\mathit{equal}=x_1 \approx x_2$, then \code{if $\mathit{equal}(m,n)$ then $P$
  else $Q$} will execute $P$ if $m =_\ET n$ and $Q$ otherwise. (In the
  following, we will use $m=n$ as a short-hand for
  $\mathit{equal}(m,n)$).
The event construct is merely used for annotating processes and will
be useful for stating security properties. For readability, we
sometimes omit 
trailing 0 processes and else branches that consist of
a $0$ process.
The remaining constructs are used to manipulate state and 
were introduced with SAPiC~\cite{KK-jcs16}.
The construct insert $m$,$n$ binds the value $n$ to a key
$m$. Successive inserts overwrite this binding, 
the delete $m$ operation `undefines' the binding.
The construct lookup $m$ as
$x$ in $P$ else $Q$ allows for retrieving the value associated to $m$,
binding it to the variable $x$ in $P$. If the mapping is undefined for
$m$, the process behaves as $Q$. 
The lock and unlock constructs are
used to gain or waive exclusive access to a resource $m$, in the style
of Dijkstra's binary semaphores: if a term $m$ has been locked, any
subsequent attempt to lock $m$ will be blocked until $m$ has been
unlocked. This is essential for writing protocols where parallel
processes may read and update a common memory.

\begin{example}\label{ex:mon-process}
    The centralized monitor from Example~\ref{ex:mon-fst} can be
    modelled as follows. For $\NormalAction/0$, we can model the doctor's role
    as follows:
\begin{lstlisting}
D$\defeq$in(a); if a = $\NormalAction$ then
       out(<$\langle \mathsf{'Do'},a \rangle$,sign($\langle \mathsf{'Do'},a \rangle$,sk('D'))>).
\end{lstlisting}
    The centralized monitor itself verifies the signature and logs the
    access using the event construct. Note that it does not check
    whether $a$ constitutes a `normal' action.
\begin{lstlisting}
M$\defeq$ (in(<$m_1\defeq \langle \mathsf{'Do'},a\rangle$,m1s>); 
  if verify(m1s,$m_1$,pk(sk('D')))=true() then
    event LogD(a);event Execute(a))
\end{lstlisting}
To model these parties running arbitrarily many sessions in parallel,
we compose $D$ and $M$ to $!D \mid !M$.
\end{example}


As usual, the semantics are defined by means of a reduction relation.
%
A configuration $c$
consists of the set of running processes, the global store and more.
By reducing some process, it can transition into a configuration $c'$.
This relation is denoted
$c \xrightarrow{F} c'$, where the fact $F$ denotes an event, e.g.,
$\mathit{LogD}(a)$, or the adversary sending a message $m$, in which
case $F=K(m)$. 

An \emph{execution} is a sequence of related configurations, i.e., $c_1
\xrightarrow{F_1} \cdots \xrightarrow{F_n} c_{n+1}$
The sequence of non-empty facts $F_i$ defines the trace of an
execution. 
Given a ground process $P$,
$\traces(P)$ is the set of traces 
that start from an initial configuration $c_0$ (no messages
emitted yet, no open locks etc.).

%
%
To specify trace properties,
SAPiC and the 
underlying
Tamarin prover~\cite{SMCB-csf12}
support a
fragment of first-order logic
with atoms $F@i$ (fact $F$ is at position $i$ in trace),
$i<j$ (position $i$ precedes position $j$) 
and equality on positions and terms 
(see. Appendix~\ref{app:logic}
\ for its formal definition).
We write $t\vDash \varphi$  if
$t$ satisfies a trace property $\varphi$.

\subsection{Accountability protocols}\label{sec:calculus-acc}

A process by itself does not encode which of its subprocesses
represents which agent. Hence, for each agent $A$, we assign a process
$P_A$. 
Furthermore, in order to model the adversary
taking control of agents, each agent needs to specify a corruption
procedure. At the least, this corruption procedure outputs that agent's
secrets. In our calculus, these are the free names (unbound by input
or $\nu$) in $P_A$. To
model other capabilities obtained by corrupting an agent, e.g., database
access, we allow for an auxiliary process $C_A'$ to be specified.

\begin{definition}[accountability protocol]\label{def:accountability-protocol} 
    Assume a set of parties $\Agents=\set{A_1,A_2,\ldots, A_n}$ and
    $\Trusted\subseteq\Agents$.
    An accountability protocol is 
    a ground process of the following form:
    $\nu \vec a; (P_{A_1} \mid C_{A_1} \mid \cdots  \mid P_{A_n} \mid C_{A_n})$
    where  $C_{A_i}$ is of form 
    $\pevent~\Corrupt(A_i); \pout('c',\langle a_1,\ldots,a_m \rangle); C'_{A_i}$
    and $\set{a_1,\ldots,a_m} \subseteq \vec a$
    are the free names in $P_{A_i}$ if $A_i\notin \Trusted$, and $0$
    otherwise.
\end{definition}

\begin{example}
    The centralized monitor protocol from Examples~\ref{ex:mon-fst}
    and~\ref{ex:mon-process} is an accountability process
    \[ 
    D \mid (\pevent~\Corrupt('D'); \pout(\sk('D')) \mid M.
    \]
\end{example}


Processes accessing the store can specify auxiliary
processes $C_{A_i}$, as
per default, SAPiC does not permit the
adversary to emit events or access the store. 
\begin{full}
E.g., to give a read access for some deviating party $E$ to selected
    parts of the store, we define
\begin{lstlisting}
$C_{E}'\defeq$ !(in('c',<x,y>); insert <'Postbox',x> y)
    | !(in('c',x); lookup <'Postbox',t> as y in out('c',y)).
\end{lstlisting}
We leave this auxiliary process informal, as it depends on how the
store is used.
\end{full}
With these formal requirements, we can  define the set of corrupted
parties of a trace as
$\corrupted(t)=\set{A\in\Agents \mid \Corrupt(A)\in t}$.

The accountability mechanism is defined through the accountability
protocol itself and the verdict function. We require the verdict
function to be invariant w.r.t.\ $\ET$.
\begin{example}\label{ex:mon-verdict}
    The verdict function for
    the centralized monitor protocol from Examples~\ref{ex:mon-fst}
    and~\ref{ex:mon-process}, which we sketched in
    Example~\ref{ex:mon-coarse}, can be specified as:
    \[
        \verdict(t) = \begin{cases}
            \set{\set{D}}& \text{if $t\vDash \omega_1$} \\
            \emptyset & \text{if $t\vDash \omega_2 $} \\
        \end{cases}
    \]
    for $\omega_1\defeq \exists a,i,j. \Execute(a)@i \land \mathit{LogD}(a)@j \land a\neq \NormalAction$
    and $\omega_2 \defeq \neg\omega_1$.
    As only $\set{\set{D}}$ is to be blamed in this example,
    this verdict function achieves 
    accountability for
    \[ 
    \varphi \defeq \forall a',i'.\ \Execute(a')@i' \implies
            a = \NormalAction, \]
    even w.r.t.\ the
    weakest relation $r_w$.
    This can be shown automatically by verifying the verification conditions for $r_w$
    (Def.~\ref{def:verification-conditions-coarse}), i.e.:
    \begin{itemize}
        \item exhaustiveness and exclusiveness,
        \item sufficiency of $\set{\set{D}}$
            $(\SufficiencyB[1][\set{\set{D}}])$:
            there is a trace $t$ s.t.\ 
            $t\vDash \exists i.
            \Corrupt(D)@i \land \neg (\forall j,a.\ \Execute(a)@j
            \implies a = \NormalAction)$, i.e.,
            a corrupt $D$ is able to execute an exceptional action.
        \item verifiability w.r.t.\ $\omega_1$
            ($\Verifiability[1]$):
            $\omega_1 \implies \neg \varphi$, which holds a~priori.
        \item verifiability w.r.t.\ $\omega_2$
            ($\Verifiability[2]$):
            $\omega_2 \implies \varphi$, i.e.,
            the absence of a log for an
            exceptional action  means none was effectuated.
        \item minimality of $\set{\set{D}}$ ($\Minimality[1]$):
            $\varphi \lor 
            \exists i.\allowbreak  \Corrupt(D)@i$.
            Unless $D$ was corrupted, no exceptional action was
            effectuated.
        \item uniqueness of $\set{\set{D}}$ ($\Uniqueness[1]$):
            $\omega_i \implies 
            \exists i.  \Corrupt(D)@i$. An entry in the log
            blaming $D$ can only occur if $D$ was actually corrupted.
    \end{itemize}
\end{example}



\section{Case studies for $r_w$}\label{sec:casesrw}\label{sec:implementation}

In this section and Section~\ref{sec:casesrc},
we demonstrate the feasibility of our verification
approach on various case studies in different settings.
We first concentrate on cases where the weakest counterfactual
relation $r_w$ is sufficient, including practical examples like
Certificate Transparency and OCSP-Stapling.

We implemented our translation from accountability properties to
conditions in SAPiC\footnote{Currently available in the
development branch of Tamarin and to be included in the next release: 
\url{https://github.com/tamarin-prover/tamarin-prover}.},
which provides support for
arbitrary relations (leaving the proofs for $\RelationLifting$ and
$\RelationReflexive$ to the user), the relation $r_c$ (as described in
Section~\ref{sec:ctl-flow}) and the weakest possible relation $r_w$
(Section~\ref{sec:cond-coarse}).
Our fork retains full compatibility with the classic SAPiC semantics,
with the extension for liveness properties~\cite{BaDrKr-2016-liveness}
and operates without any substantial changes to Tamarin.
\begin{full}
(We have only optimized the SAPiC-specific heuristic, which is part
of Tamarin, to deprioritize goals pertaining to annotations
introduced with the sequential self-composition in
Section~\ref{sec:ctl-flow}.)
\end{full}
By default, our fork preserves multiset rewrite rules contained in its
input, and can thus also serve as a preprocessor for accountability
protocols encoded in Tamarin's multiset rewriting calculus.



Our findings are summarized in Table~\ref{tab:studies}.
For each case study, we 
give the type (\checkmark{} for 
successful
          verification, \xmark{} if we discovered an attack),
the number of lemmas generated by
      our translation,
the number of additional helping lemmas\footnote{%
SAPiC, as well as tamarin, are sound and complete, but the underlying
verification problem is undecidable~\cite{AbadiCortierTCS06}.
Therefore, analyses in SAPiC/Tamarin sometimes employ \emph{helping
lemmas} to help the verification procedure terminate. Just like security
properties, they are stated by the user and verified by the tool.
}
and the time
      needed to verify \emph{all} lemmas (even if an attack was found).
Verification was performed on a 2,7 GHz Intel Core i5 with 8GB RAM.

\begin{table}[t]
  \centering
      \caption{Case studies and results.}\label{tab:studies}

    \begin{tabular}{lcccr}
        & & \# lemmas & \# helping & \\
        protocol & type & generated &  lemmas & time\\
        \toprule
        Whodunit & \\
        \quad faulty & \xmark{},$r_w$ &\ 16 & \ 0 &\ 395s \\
        \quad fixed & \checkmark{},$r_w$ &\ 8 & \ 0 &\ 112s  \\
         \midrule
        Certificate Transp. & \\
        \quad model from~\cite{DBLP:conf/isw/BruniGS17} & \checkmark{},$r_w$ &\ 31 & \ 0 &\ 41s \\ 
        \quad extended model & \checkmark{},$r_w$ &\ 21 & \ 0 &\ 50s \\ 
        \midrule
        OCSP Stapling & \\
        \quad trusted resp. & \checkmark{},$r_w$ &\ 7 & \ 3 &\ 945s \\ 
        \quad untrusted resp. & \xmark{},$r_w$ &\ 7 & \ 3 &\ 12s$^*$ \\ 
        \midrule
        Centralized monitor & \\
        \quad faulty & \xmark{},$r_c$ &\ 17 & \ 0 &\ 5s \\ 
        \quad fixed & \checkmark{},$r_c$ &\ 17 & \ 0 &\ 3s \\ 
        \quad replication & \checkmark{},$r_c$&\ 17 & \ 0 &\ 7s \\ 
        \midrule
        Causality & \\
        \quad Desert traveller & \checkmark{},$r_c$ &\ 16 & \ 0 &\ 7s \\ 
        \quad Early preempt. & \checkmark{},$r_c$ &\ 16 & \ 0 &\ 1s \\
        \quad Late preempt. & \checkmark{},$r_c$ &\ 16 & \ 0 &\ 13s\\ 
        \midrule
        Accountable alg. & \\
        \quad modified-1 & \checkmark{},$r_c$ &\ 27 & \ 1 &\ 5792s \\ 
        \quad modified-2 & \checkmark{},$r_c$ &\ 27 & \ 1 &\  2047s \\         
        \bottomrule
    \end{tabular}
    \vskip 0.1in
    (*): time to find counter-examples for $\Verifiability$ and $\Uniqueness$.
\end{table}

\subsubsection*{Toy example: whodunit}\label{sec:whodunint}

\begin{figure}
\centering
\begin{lstlisting}
S := in (a); out($c_{SA}$,a); out($c_{SJ}$,a)
A := in($c_{SA}$,a); out ($c_{AJ}$,a)
J := in($c_{SJ}$,$a_1$); in($c_{AJ}$,$a_2$); if $a_1$=$a_2$ then event Equal()
                                    else event Unequal()
new $c_{SA}$; new $c_{AJ}$; new $c_{SJ}$; (S || A || J ||
 !(in(<'corrupt',x>); event Corrupted(x); out(sk(x));
     (if x='S' then out($c_{SA}$); out ($c_{SJ}$))
   ||(if x='A' then out($c_{SA}$); out ($c_{AJ}$)))))
\end{lstlisting}
\caption{Whodunit protocol\cite[Ex.~8]{cryptoeprint:2018:127}.}
\label{fig:whodunit-protocol}
\end{figure} 

The example in Figure~\ref{fig:whodunit-protocol} illustrates the
difference between verdicts larger than 2 and uncertainty about the
correct verdict.
Two parties, $S$ and $A$,
coordinate on some value chosen by $S$. $S$ sends this value to $A$ 
and to a trusted party $J$. Then, A is supposed to forward this value
to $J$.
We are interested in accountability for $J$ receiving the same value
from $S$ and $A$.

The crux here is that
a correct verdict function cannot exist: 
if $J$ receives different values, 
there are two minimal explanations for her. Either
$A$ altered the value it received or $S$ gave $A$ a
different value in the first place.
Indeed, if we formalize this in two verdict functions, one blaming $A$
if the fact $\mathfun{Unequal}$ occurs in the trace, the other blaming
$S$, Tamarin finds a counterexample for each. 
If we change the protocol so that $S$ needs to sign her message, and
$A$ needs to forward this signature, then we can prove accountability
for the verdict function that blames $S$ in case of inequality. 

\subsubsection*{Certificate Transparency}

Certificate Transparency~\cite{rfc6962} is a protocol that provides
accountability to the public key infrastructure. Clients are
submitting certificates signed by CAs to logging authorities (LAs), who
store this information in a Merkle tree. Auditors validate that these
logs have not been tampered with. Based on these trustworthy, distributed
logs, clients, e.g., domain owners, may detect misbehavior, e.g., an
unauthorised CA, issuing certificates for this domain.

We base our modelling on that of Bruni et.al.~\cite{DBLP:conf/isw/BruniGS17},
which considers a simplified setting with one CA, one LA, and two
honest auditors. We first verify accountability of the CA for the
property that any certificate in the log that is tested was honestly
registered. To this end, we, as well as Bruni et.al., have to assume
access to the CA's domain registration data. We then verify
accountability of the LA for the property that any log entry that was
provided was provided consistently to all auditors. Finally, we
compose both security properties and verdict functions, and can thus
show that CA and LA can be held accountable at the same time.

While the original modelling prescribed cheating LAs to cheat in
a certain way (always provide the correct log entry to auditor $u_1$
and omit it to auditor $u_2$), we extended the model to permit
deviating LAs to selectively provide log entries.
This complicates the formulation of the consistency criterion, but
makes the model slightly more realistic.
Both models can be verified within a minute.

\subsubsection*{OCSP Stapling}

\begin{figure}
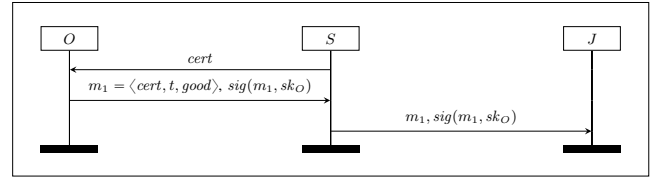

\centering
    \resizebox{\linewidth}{!}{
    \setmscvalues{small}
    \begin{msc}[instance distance = 4.3cm,level height = 0.65cm,
        msc keyword = {},
        head top distance = 5mm,
        foot distance = 5mm,
        ] 
        {}
    \declinst{O}{}{$O$}
    \declinst{S}{}{$S$}
    \declinst{J}{}{$J$}
 \mess{$\mathit{cert}$}{S}{O}
\nextlevel
        \mess{$m_1=\langle \mathit{cert},t,\mathit{good}\rangle$, $\sig(m_1,\sk_O)$}{O}{S}
\nextlevel
        \mess{$m_1,\sig(m_1,\sk_O)$}{S}{J}
\end{msc}
}
    \caption{OCSP Stapling}
\label{fig:ocsps}
\end{figure} 

The Online Certificate Status Protocol (OCSP~\cite{rfc6960}) provides an interface to
query whether a digital certificate is revoked, or not. Upon a request
(which may be signed or not), a trusted responder (e.g., a CA or
a delegate) gives a signed response containing a timestamp, the certificate in
question and its revocation status (\emph{good}, \emph{revoked} or
\emph{unknown}). \emph{OCSP Stapling}~\cite{rfc6066}, is an extension
to TLS that specifies how a TLS server may attach
a recent enough OCSP response to a handshake. This reduces the
communication overhead. In addition, it avoids clients 
exposing their browsing behavior to the OCSP server via their
requests.

We model OCSP stapling as an accountability protocol between a trusted
OCSP responder $O$, an untrusted server $S$ and a trusted Judge $J$.
The judge represents a client that receives a stapled response from
the server and seeks to determine if its communication partner can be
trusted.
In addition, a clocking process emits timestamps.
Our modelling is quite simplistic, e.g., the TLS Handshake is reduced
to a forwarding of the signed timestamp.
The main challenge we focussed on was
defining the accountability property that is actually achieved.
First, note that a server can choose to reveal its secret
key at any time.
In order to make any meaningful statement about the revocation mechanism,
we have to limit ourselves to cases where, whenever a server reveals
his secret, it also revokes the corresponding certificate.
We thus
slightly diverge from the corruption procedure in
Definition~\ref{def:accountability-protocol}, and require the server
to mark his certificate as revoked upon corruption. 

We can show accountability for $\varphi=$
\newcommand{\ski}{\mathfun{sk}}
\begin{align*}
    \neg \exists c,\ski,t,i,j,k,l.&
    \mathit{Judged}(t,c)@i \land \mathit{Secret}(c,\sk)@j \\
    & \land K(sk)@k \land \mathit{Time}(t)@l \land k < l,
\end{align*}
i.e., 
whenever a client
received an OCSP response for a certificate $c$ with timestamp $t$ (at
which point $\mathit{Judged}(t,c)$ is raised),
she
can be assured that the corresponding secret $\sk$ was not leaked
(recall that $K(\sk)$ marks the adversary sending a message) 
at a point in
time $k$ prior to the emission of the timestamp $t$ at time $l$ (but possibly later).
Timestamps are here modelled as public nonces emitted by the clock
process. Prior to outputting a nonce, the event $\mathit{Time}(t)$ is
emitted, binding these nonces to positions in the trace.

To explore the limits of OCSP, we declare the OCSP responder to be
untrusted. We find an attack 
on the previous accountability property
where $O$ reports a \emph{good} certificate status to $J$, despite the
revocation triggered when corrupting $S$.
The techniques used by Milner et.al.\ to detect misuse of secrets~\cite{8049721}
could potentially be used to mitigate this issue. We leave this
research question for future work.

\section{Verifying RL and RS}\label{sec:ctl-flow}

For protocols in which parties can collude, but can also cause damage
on their own, we need to verify accountability with respect to
a relation $r$ between what actually happened, and what could have
happened with a subset of these parties (see
Example~\ref{ex:mon-snd}).
The verification conditions in Section~\ref{sec:ver} provide
a framework for many counterfactual relations $r$. 
%
%
In this section, we show how $r$ can be instantiated to 
the relation $r_c$ 
discussed 
in Section~\ref{sec:relations}, 
so that the conditions
$\RelationLifting$ and $\RelationReflexive$, and thus accountability
as a whole, can be verified with SAPiC and Tamarin.
\begin{conf}
For other relations, which includes relations that cannot be formally
stated such as $r_k$, $\RelationLifting$ and $\RelationReflexive$ need
to be proven or justified on paper (but the remaining conditions can be
verified).
\end{conf}


In practice, the control-flow of a process is not necessarily the
control-flow of its implementation. To leave some degree of
flexibility to the modeller, we allow for the control-flow to be
manually annotated and use a unary fact $\Control\in\FSign$ to
mark control-flow. Per default, there should be exactly one statement
$\pevent~\Control(p:\pubsort)$ on each path from the root to a leaf of the
process tree of each process corresponding to a trusted party.
We can then define $r_c(t,t')$ iff for all $p$ and $p'$ s.t.\ 
$\Control(p)\in t$
and
$\Control(p')\in t'$, $p=_\ET p'$.

The main challenge in proving $\RelationLifting$ and $\RelationReflexive$
is that SAPiC supports only trace properties.\footnote{Tamarin
supports diff-equivalence~\cite{BDS-ccs15}. This variant of
observational equivalence considers two processes unequal if they move
into different branches, hence it is not suitable for our case.}
Hence, in general, we can argue about all or some $t$, but not 
about pairs of $t$ and $t'$.
The solution is to combine $t$ and $t'$ in
a single trace,
which is their concatenation.
If for all occurrences of $\Control(p)$ in the first part,
and of
$\Control(p')$ in the second part, $p$ and $p'$ coincide,
then the $t$ and $t'$ corresponding to these two parts are in the
relation.
This technique is known as self-composition~\cite{barthe2004secure}.

Defining this sequential composition of $P$ with itself is technically
challenging and requires altering the translation from SAPiC to
multiset rewrite rules --- observe that $P;P$ is not a syntactically
valid process. Due to space limitations, we refer to
Appendix~\fvref{sec:ctl-full} and~\fvref{sec:ctl-proof}
for the technical solution and its proof of correctness (it is
sound and complete), and will only discuss the idea and limitations.

The idea is that the adversary can start executions with a fresh
execution id, which is used to annotate visible events. To separate
executions, the adversary can trigger a stopping event. A rewriting of
the security property ensures that it is trivially true,
unless both executions are properly separated, i.e., every execution
is terminated, events are enclosed by start and stop events, and these
events themselves define disjoint intervals.

Note that for a process where a trusted party is under replication, it
is possible that two different $\Control()$-events are emitted in the
same execution, and thus the corresponding process is in no
equivalence class w.r.t.\ $r_c$.
This affects only one of our case studies (centralized monitor),
however, instead of considering a possibly replicating series of
violations, we chose to identify the party causing the \emph{first}
violation.
%
Appendix~\fvref{sec:ctl-full}
discusses other possible solutions to this issue in more detail.


\section{Case studies for $r_c$}\label{sec:casesrc}

The most challenging protocols from an accountability perspective are
those in which joint misbehaviour is possible. 
The centralized monitoring mechanism from Example~\ref{ex:mon-snd}
provides such an example, as well as the accountable algorithm's
protocol proposed by Kroll.
In both cases, an analysis w.r.t.\ a more restrictive
counterfactual relation is strictly necessary. We opt for the relation
$r_c$, relating runs having the same control-flow. To this end, we use
the elaborate verification condition for arbitrary $r$ (Def.~\ref{def:ver-conditions})
and automate the analysis of $\RelationLifting$ and
$\RelationReflexive$ by considering the sequential self-composition
of the protocol as described in the previous section.
Owing to our accountability definition's origins in causation, we will
discuss examples of `preemption' from the causation literature,
which are considered difficult to handle, but can be tamed by
considering the control-flow of execution.
Note that we omit code listings for most examples, however, they come
with the implementation.

\subsubsection*{Centralized monitor}

Example~\ref{ex:mon-snd} considered a protocol based on a central
trust monitor $M$.
Albeit modelled in a very abstract form (the actual
protocols are likely to be tailored for their use case), this kind of
mechanism occurs in various forms in plenty of real-world scenarios.
We want to demonstrate that, in principle, we can handle such
scenarios.

A party $D$ can effectuate actions,
some of which are usual (e.g., a doctor requesting his patient's
file),
some of which are not (e.g., requesting the file of another doctor's
patient).
Rather than blocking access for exceptional action, these are logged by
$M$,
(e.g., if another doctor's patient has a heart attack and needs
treatment right away), which is an accountability problem w.r.t.\ the
property that no exceptional action happened.
A supervisor $C$ (e.g., chief of medicine) is needed to effectuate a third
class of actions, exceptional actions, for which $D$ needs to get $C$'s
authorization. The processes of parties $D$, $C$ and $M$ are running
in parallel with a process that outputs their public keys and 
their private keys on request (see Fig.~\ref{fig:delegation-protocol}).

\begin{figure}
\centering
\begin{lstlisting}
D$\defeq$in(a);
   if isNormal(a) then
       out(<$m_1\defeq \langle \mathsf{'Do'},a \rangle$,sign($m_1$,sk('D'))>)
   else if isExceptional(a) then
        out(<$m_2\defeq\langle\mathsf{'Permit'},a\rangle$,sign($m_2$,sk('D'))>)) 
C$\defeq$in(<$m_2\defeq\langle\mathsf{'Permit'},a\rangle$,m2s>);
     if verify(m2s, $m_2$, pk(sk('D')))= true() then
        if isExceptional(a) then
          out(<$m3\defeq\langle m_2, m2s\rangle$,sign($m3$,sk('C'))>)
M$\defeq$
 (in(<$m_1\defeq \langle \mathsf{'Do'},a\rangle$,m1s>); 
  if verify(m1s,$m_1$,pk(sk('D')))=true() then
  event Control('0','1');event LogD(a);event Execute(a))
+(in(<$m3\defeq\langle m_2, m2s\rangle$,m3s>); // for $m_2\defeq\langle\mathsf{'Permit'},a\rangle$
  if verify(m3s, $m3$, pk(sk('C'))) = true() then
  if verify(m2s, $m2$, pk(sk('D'))) = true() then
  event Control('0','2');event LogDC(a);event Execute(a))

!($D$ || $C$) || $M$
|| //give access to public keys
   (out(pk(sk('D')));out(pk(sk('C')));out(pk(sk('M'))))
||!(in('c',<'corrupt',x>); 
    event Corrupted(x); out('c',sk(x))))
\end{lstlisting}
\caption{Centralized monitor.}
\label{fig:delegation-protocol}
\end{figure} 

We use function symbols 
$\NormalAction/0$, $\ExceptionalAction/0$ 
to denote these kinds of actions, and 
$\sign/2$, $\verify/2$ to model digital signatures.

$D$ receives an action $a$ from the adversary (for
generality) and either signs it and sends it to $M$ directly (if $a$
is `normal'), or signs a permission request, which $C$ has to sign
first (if $a$ is `exceptional').
$C$ only signs requests for exceptional actions. $M$ non-deterministically
guesses what kind of message arrives, and verifies that the signatures
are correct. If this is the case, the action is executed.

We investigate accountability for the property that only
`exceptional' or `normal' actions are executed:
\begin{align*}
    \forall a,i. & \Execute(a)@i \\
    & \implies a=\ExceptionalAction() \lor a=\NormalAction() 
\end{align*}

Within 3 seconds, our toolchain shows that the verdict function that maps the
occurrence of $\mathname{LogD}$ to $\set{\set{D}}$ and the occurrence of
$\mathname{LogDC}$ to $\set{\set{D,C}}$ provides accountability for this
property.

\begin{full}
\begin{figure}
\centering
\begin{lstlisting}
A$\defeq$ 0 //Poisoned -- by default, no poison
B$\defeq$ 0 //Shot -- by default, no shooting
V$\defeq$  //Victim
             (in(~p); insert 'Victim',~p; in(~s); insert 'Victim',~s) 
          +  (in(~s); insert 'Victim',~s) 
C$\defeq$ insert 'Victim', 'Healthy'; 
     lookup 'Victim' as x in (
        if x=~s then (
          event Shoot(); event Control('0','1'); event Verdict()
          )
        else (
          if x=~p then (event Poisoned(); event Control('0','2');event Verdict())
          else ( event Healthy(); event Control('0','3'); event Verdict())
        )
     )

new ~p; new ~s; (A || B || C || V ||
 !(in ('c',<'corrupt',x>); event Corrupted(x); 
     !((if x='A' then out('c',~p)) || (if x='B' then out('c',~s)))))
\end{lstlisting}
\caption{Late preemption.}
\label{fig:late-preemption}
\end{figure} 
\end{full}

The first protocol design was (without intention) erroneous.  We find
two attacks
within 5 seconds (for
falsification and verification of all lemmas). 
\begin{full}
The first attack results from a confusion of message $m_1$ and $m_2$ (solved
here via tagging). The second results from having $m_3$ rely on a nonce
to bind it to $m_2$, which is vulnerable to a corrupted $A$
signing a different action, with the same nonce, later (we solve this
by including $m_2$ in $m_3$). 
Both manifest as a violation of
$\Minimality$ and $\Uniqueness$ in their respective cases,
since a strict subset of the faulty verdicts $\set{A}$, respectively $\set{A,B}$,
is enough to trigger the case.
\end{full}

For simplicity, $M$ is not covered under replication. Putting $M$ under
replication, it is possible to have two different
$\Control$-events in the same run, which entails that no second
trace can relate via $r_c$.
We can, however, prove a modified property $\varphi'$ which is true only if
there are no violations to $\varphi$, or \emph{at least} two. 
Intuitively, this means that we can point out which party caused the
\emph{first} violation by considering, for every violating trace, the prefix
which contains only one violation.
This is possible for all safety-properties, as they are, by
definition, prefix-closed~\cite{lamport1977proving}.
In addition, we modify $M$ to only emit $\Control$-events when acting
upon an action that is neither normal nor exceptional, i.e., we consider
only control-flow for actions which produce violations.

\subsubsection*{Examples from causation literature}

As our accountability definition is rooted in causation, we chose to
model three examples from the causation literature, two from
Lewis' seminal work on
counterfactual reasoning~\cite{Lewis1973-LEWC}, and one 
formulated by Hitchcock~\cite[p.
526]{Hitchcock2007-HITPPA}.
They all encode problems of \emph{preemption}, where
a process that would cause an event (e.g., a violation)
is interrupted by a second processes.
All examples are verified in a couple of seconds.
We refer to
Appendix~\fvref{sec:ex-causation} for details.

\subsubsection*{Kroll's Accountable Algorithms protocol} 

The most interesting case study is the accountable
algorithms protocol of Kroll~\cite[Chapter~5]{krollphd}. 
%
It lets an authority $A$, e.g., a tax authority, perform
computations for a number of subjects  $S_1$ to $S_n$, e.g., the
validation of tax returns.
Any subject can verify, using a public log, that the value it receives
has been computed correctly w.r.t.\ the input it provides and some
secret information that is the same for all parties.
We substantially extend this protocol to also provide accountability
for the subjects: if a subject decides to falsely claim that the
authority cheated, we can validate their claim.
This is a very instructive scenario, as now any party emitting logs or
claims can just lie. It also demonstrates that a trusted \emph{third}
party is not needed to provide accountability. While we define
a judgment process $J$, which is technically a trusted party 
(it is always honest), this party
is (a) not involved in the protocol and (b) can be computed by
anyone, e.g., a journalist, an oversight body, or the general public.

\begin{figure}
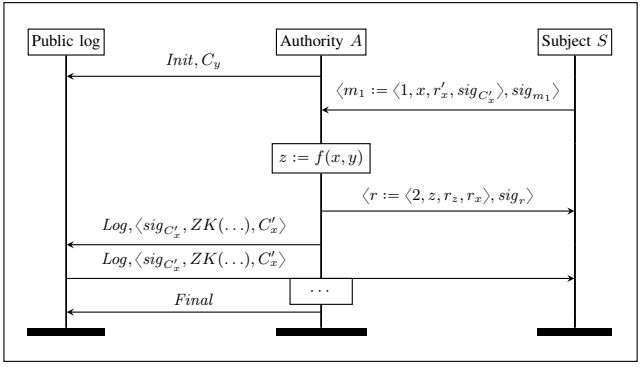

\centering
    \resizebox{\linewidth}{!}{
    \setmscvalues{small}
    \begin{msc}[instance distance = 3.3cm,level height = 0.65cm,
        msc keyword = {},
        head top distance = 5mm,
        foot distance = 5mm,
        ] 
        {}
    \declinst{L}{}{Public log}
    \declinst{A}{}{Authority $A$}
\declinst{S1}{}{Subject $S$}
 \mess{$\Init,C_y$}{A}{L}
\nextlevel
\mess{$\langle m_1 \defeq \langle 1, x, r_x', \sig_{C_x'} \rangle, \sig_{m_1}  \rangle$}{S1}{A}
        \nextlevel
        \action*{$z\defeq f(x,y)$}{A}
        \nextlevel[2]
        \mess{$\langle r \defeq \langle 2, z, r_{z},r_{x} \rangle, \sig_{r} \rangle$}{A}{S1}
\nextlevel
        \mess{$\Log,\langle \sig_{C_x'}, ZK(\ldots), C_x' \rangle$}{A}{L}
\nextlevel
        \mess{$\Log,\langle \sig_{C_x'}, ZK(\ldots), C_x' \rangle$}{L}[.25]{S1}
        \action*{$\cdots$}{A}
\nextlevel
        \mess{$\Final$}{A}{L}
\end{msc}
}
    \caption{Honest protocol run for
    accountable algorithm protocol.}
\label{fig:honest-protocol-run-for-example-ex}
\end{figure} 

\begin{figure}[!htb]
\centering
 \begin{lstlisting}
S$\defeq$ let res     = <'2', z, rz, rx>
      claim   = <'3', x, z, res, sig_res, rxp>
      Cxp     = commit(x, rxp)
      m1      = <'1', x, rxp, sign(Cxp, sk('S'))>
      sig_cxp = fst(log)
      zkp     = fst(snd(log))
      CxpL    = snd(snd(log))
      Cy      = fst(snd(Pub(zkp)))
  in 
  lookup <'A', 'Init'> as CyL in 
  new x; new rxp; out (<m1, sign(m1,sk('S'))>); 
  lookup <'A','Log','S'> as log in
  if and4( verZK(zkp), eq(sig_cxp,sign(Cxp, sk('S'))),
           eq(CxpL,Cxp), eq(Cy, CyL)) then 
                in (<res, sig_res>); 
                if verify(sig_res, res, pk(sk('A')))=true() then
                  out (<claim, sign(claim, sk('S'))>)

A$\defeq$ let m1 = <'1', x, rxp, sig_cxp>
      z      = f(x, y)
      res    = <'2', z, rz, rx>
      Cx         = commit(x, rx)
      Cy         = commit(y, ry)
      Cz         = commit(z, rz)
      Cxp        = commit(x, rxp)
      zkp        = ZK(<Cx,Cy,Cz>,x,y,z,rx,ry,rz)
      sig_res    = sign(res, sk('A'))
  in
  new y; new ry;
  insert <'A', 'Init'>, Cy;
  in (<m1, sig_m1>); new rx; new rz;
  if and3(verify(sig_m1,m1,pk(sk('S'))), 
      eq(x, open(Cxp,rxp)), 
      verify(sig_cxp,Cxp,pk(sk('S')))) then 
      out (<res, sig_res>);  // send result to S
      insert <'A','Log','S'>, <sig_cxp, zkp, Cxp>;
      insert <'A','Final'>, true()

!($A$ || $S$) || $J$ ||
  out(pk(sk('A'))); out(pk(sk('S')))
||  
 !(in('c',<'corrupt',x>);event Corrupted(x);out('c',sk(x));
   !((if x='A' then in(y); insert <'A','Init'>,y)
    ||(if x='A' then in(y); insert <'A','Log','S'>, y)
    ||(if x='A' then in(y); insert <'A','Final'>, y)) 
    ||(if x='A' then lookup <'A','Log','S'> as y in out(y))
    ||(if x='S' then lookup <'A','Log','S'> as y in out(y))
    ))
 \end{lstlisting}
\caption{Accountable algorithms protocol.}
\label{fig:accountable-algorithms-protocol}
\end{figure} 

The goal is to compute a function $f$ on inputs $x$ and $y$; $f$
representing an arbitrary algorithm, $x$ the subject's input, and $y$
some secret policy that $f$ takes into account (e.g., an income
threshold to decide which tax returns need extra scrutiny). The policy
$y$ is the same for all subjects.
First, the authority commits on $y$ (see
Figure~\ref{fig:honest-protocol-run-for-example-ex}).
We model commitments using the symbols
$\commit/2$,
$\open/2$, $\verCom/2$ and equations
$\open(\commit(m,r),r)=m$
and
$\verCom(\commit(m,r),r)=\true()$.
Next, the subject sends its input $x$ along with randomness for its
desired commitment $r_x'$ and a signature for the commitment
$C_x'\defeq \commit(x,r_x')$.
Now the authority computes $z=f(x,y)$ and returns $z$, as well as the
randomness which was used to compute two commitments on its own, $C_x$ on $x$
and $C_z$ on $z$.
The signed commitment $C_x'$ and a zero-knowledge proof (ZKP) are
stored in a public append-only log, e.g., the blockchain. 
The log is modelled via the store and a global assumption that entries
cannot be overwritten.
The ZKP contains the three commitments $C_x$, $C_y$ and
$C_z$ and shows that they are commitments to $x$, $y$ and $z$ such
that $f(x,y)=z$. Using this ZKP, one can check that $C_y$ is the value
initially committed to, and that the input $x$ and the output $z$
are consistent. 
Now $A$ proceeds with the next subject, and appends a $\Final$ message
to the log to indicate when it is done. 
$S$ can decide to file a claim, consisting of $x$, $z$, $r_{x}'$ and the signed
message it received from $A$ in the second step.

The original protocol was implemented with hash functions and
ZK-SNARKs.
Kroll's preliminary analysis~\cite[Section~5.2.2]{krollphd} was
informal and only considered holding the authority accountable.
Later analysis discovered
that any subject can falsely accuse the authority $A$ of
misbehavior~\cite{cryptoeprint:2018:127}.
Obviously, such a claim would subvert the trust into the system and
hence render the accountability mechanism useless in practice. 
Consequently, we extended the protocol with a message authentication
mechanism based on digital signatures so that $A$ and $S_i$ prove
their identity and protect themselves from false accusations, as well
as $C_x'$, which is a second commitment on $x$ generated by $S$
instead of $A$.

This commitment (along with $S$'es signature) serves as a witness
that
the input which $A$ claims to have received is indeed the one $S$ has sent,
without revealing $x$ to the public.

On the other hand, we simplified Kroll's protocol by removing
randomness generation for $f$, which was implemented using verifiable
random functions. We outline in Section~\ref{sec:concl} how this
feature calls for a definition of accountability in the computational
setting.

We model this protocol as an accountability process with three
subprocesses $A$, $S$ and
a judging procedure $J$ (see
Figure~\ref{fig:accountable-algorithms-protocol}
and~\ref{fig:accountable-algorithms-protocol-judging-procedure}).
This judging procedure essentially determines
the verdict by inspecting the log and a claim output by $S$. 
For brevity, we introduce a symbol $\mathfun{eq}/2$ with equation
$\mathfun{eq}(x,x)=\true$ and predicates $\mathfun{and3}$ and
$\mathfun{and4}$ which evaluate to true if all three or four terms
equal the constant $\true$.
We also use syntactic sugar $\mathsf{let}~v = t~\mathsf{in}~P$
to denote the literal substitution of $v$ by $t$ in $P$.

\begin{figure}[t]
\centering
\begin{lstlisting}
J$\defeq$ let res     = <'2', z, rz, rx>
      claim   = <'3', x, z, res, sig_res, rxp>
      sig_cxp = fst(log)  // log is <sig_cxp, zkp, Cxp>;
      zkp     = fst(snd(log))
      CxpL    = snd(snd(log))
      Cx      = fst(Pub(zkp))
      Cy      = fst(snd(Pub(zkp)))
      Cz      = snd(snd(Pub(zkp)))
in
in (<claim, sig_claim>);
if verify(sig_claim, claim, pk(sk('S'))) = true() then
  lookup <'A','Final'> as y in 
  event Final();
  lookup <'A','Init'> as CyL in
  event Comm(CyL);
  lookup <'A','Log','S'> as log in
  // first check validity of the log by itself,
  if and3(verZK(zkp), eq(CyL, Cy), //produced by A
    verify(sig_cxp,CxpL,pk(sk('S')))) //verified by A 
  then 
      if and3(verify(sig_res,res,pk(sk('A'))), // honest S verifies this
              verCom(CxpL,rxp), // produced by S
              eq(x, open(CxpL,rxp)))) // both signed by S 
      then // We now believe S is honest and its claim valid
        if and4(verCom(Cx,rx), verCom(Cz,rz),
                eq(x, open(Cx, rx)),eq(z, open(Cz, rz)))
        then 
          event Control('0','1');event HonestS();
          event HonestA();event Verified(claim)
        else 
          if oracle(x,Cy,z)=true() then // see below.
            event Control('0','2');event HonestS();
            event HonestA(); event Verified(claim)
          else 
            event Control('0','3');event HonestS();
            event DisHonestA(); event Verified(claim)
      else  // A's log is ok, but S is definitely cheating
        if oracle(x,Cy,z)=true() then
          event Control('0','4');event HonestS();
          event HonestA(); event Verified(claim))
        else 
          event Control('0','5');event DisHonestS();  
          event HonestA();event Verified(claim)
  else (// A is dishonest and produced bad log
    if oracle(x,CyL,z)=true() then
      event Control('0','6');event HonestS();
      event HonestA();event Verified(claim)
    else
      event Control('0','7');event DisHonestA();
      event DisHonestS();
      // S checks log before submitting claim$\rightarrow$ S dishonest
      event Verified(claim))
\end{lstlisting}
\caption{Accountable algorithms protocol: judging procedure.}
\label{fig:accountable-algorithms-protocol-judging-procedure}
\end{figure} 

The verdict function is a trivial conversion of the events emitted by
the judging procedure $J$, which inspects the log and evaluates the
claim for consistency (see
Figure~\ref{fig:accountable-algorithms-protocol-judging-procedure}).
We decided for this approach in order to ensure that the judgement can
be made by the public, and to provide an algorithmic description on
how to do so, since a verdict function can easily be written to rely on
information that is not available or not computable.
$J$ is technically a trusted party in the sense that it is always honest.
However, J is not involved in the protocol and the fact can be
computed after the fact by anyone who receives the claim, e.g.,
a newspaper, an oversight body, or the general public.
A dishonest $J$ merely means that the verdict function is computed
incorrectly, in which case we cannot, and should not, make any guarantees.
The verdict function maps traces in which
$\DishonestA$ and $\DishonestS$ appear to be $\set{\set{A,S}}$,
and traces with
$\DishonestA$ and $\HonestS$ to be $\set{\set{A}}$.
It provides the protocol with accountability for the following property
$\varphi$:
\begin{align*}
    \forall t,x,z,r,\sig_r,r_{x'}. & 
    \Verified(\langle t,x,z,r,\sig_r,r_{x'} \rangle)@k 
    \\
     \implies \exists j,i,y,r_y. &
    \Final()@i \land \CommEv(\commit(y,r_y))@j 
    \\
    & \land z = f(x,y) \land
    i<j<k.
\end{align*}

This property guarantees that any claim considered by the judging
procedure is correctly computed w.r.t. $f$ and the initial commitment
to $y$. Note that a violation requires a $\Verified$-event, which $J$
only emits if, and only if, the first conditional and the two
subsequent lookups are successful.
The conditional formulates the requirement of a claim to be signed by
$S$, the two lookups require $A$ to indicate the
protocol is finished and to have produced an Init entry, in order to
ensure that the claim is only filed after the protocol has finished.

The judging procedure needs to rely on an external oracle to determine
whether a violation actually took place. This is a restriction, as it
requires any party making a judgement to have access to such an
oracle.
We need this for the following reason: accountability implies
verifiability, hence, even if the logs have been tampered with, and
they do not contain usable information, the verdict function needs to
nevertheless output the empty verdict if $z=f(x,y)$. For example, if $C_y$ in the
log does not match $x'$ and $z$ in claim, i.e., $z\neq f(x',\open(C_y,r_y))$,
but the logged $C_x$ is a commitment to
a different $x$, it could well be the case that $z=f(x,y)$.
Figuratively speaking, the adversary tries to trick us into thinking
something went wrong.
We represent the oracle as a function symbol $\oracle/3$ with equation
$\oracle(x,\commit(y,ry),f(x,y))=\true$, and restrain ourselves to
only use this function in the judging procedure.

The judging procedure is instructive in the constraints it
puts on the parties. Broadly speaking, they have to assist $J$
in holding the opposite party accountable by validating
incoming messages. E.g., if $A$ manipulates the log by
providing an invalid ZKP, $S$ terminates, and $\varphi$ is trivially
true.
This ensures that $J$ can count on the validity of the log unless $S$
is dishonest. On the other hand, $A$ can now stall the process.
From a protocol design perspective, this raises the following
challenges:
\begin{inparaenum}[\itshape 1.\upshape]
    \item Are there guiding principles for the judging procedure? (Our
        design involved a lot or trial and error, which is only
        a viable strategy due to the tool support we have.)
    \item Is it possible to achieve accountability for
        timeliness, i.e., the property that the protocol eventually
        finishes~\cite{BaDrKr-2016-liveness}?
    \item Can we do
        without oracle access in the judging procedure?
\end{inparaenum}

The verification of
the modified variant of Kroll's protocol takes
about two hours.
An alternative way of fixing this model requires a modified
zero-knowledge proof but is structurally closer to the original
proposal.
In order to avoid maintaining both $C_{x}$ (chosen by $A$)
and $C_{x'}$ (chosen by $S$), we allow $S$ to chose $C_{x}$, i.e.,
the commitment to $x$ in the public part of the ZKP, similar to how
the commitment to $y$ has to match the commitment sent out at the
beginning of the protocol. 
%
We therefore modelled and
verified both fixes, but present only one here.
This second variant can be 
verified 
in
half an hour.
In both cases, the analysis is fully automatic, but a
helping
lemma (which itself is verified automatically) was added for speedup.
It states that lookups for the same key in the log always give the
same result.

\section{Related work}
\label{sec:related}

The focus of this work is on accountability in the security setting;
here we understand accountability as the ability to identify malicious
parties.
Early work
on this subject provides or uses a notion 
of accountability in either an informal way, or tailored to the
protocol and its security~\cite{%
Asokan1998a,%
krollphd,%
DBLP:conf/wpes/BackesCS05,%
DBLP:conf/esorics/BackesFM13}.
The difficulty is defining what constitutes `malicious
behavior' and what this means for completeness, i.e., the ability of
a protocol to hold \emph{all} malicious parties accountable.
Jagadeesan et\,al.\ provided the first generalized notion of accountability,
considering parties malicious if they deviate from the protocol. 
But in their model,
`the only auditor capable of providing [completeness] is 
one which blames all principals who are capable of dishonesty,
regardless of whether they acted dishonestly or not'~\cite{Jagadeesan2009a}.
Algorithms in distributed systems,
e.g., PeerReview~\cite{Haeberlen:2007:PPA:1323293.1294279},
use this notion to detect faults in the  Byzantine setting, but need a complete
local view of all participating agents.
Our corruption model is also Byzantine, however,
we work in the security setting and
thus we do not assume that a complete
view of every component or the whole communication is available. 

For the security setting,
Küsters et\,al.\ recognize that completeness according to 
this definition of maliciousness
cannot be fulfilled (while remaining sound),
because it `includes misbehavior that is
impossible to be observed by any other party [or is harmless]'.
They propose to capture completeness via policies.
Künnemann \etal~\cite{cryptoeprint:2018:127} argue that
these policies are not expressive enough.
In case of joint misbehavior, the policy is
either incomplete, unfair, or it 
encodes the accountability mechanism itself~\cite{cryptoeprint:2018:127}.
Other approaches 
focus on protocol actions as
causes for security violations~\cite{DBLP:journals/scp/GosslerM15,%
Datta2015a,Feigenbaum:2011:TFM:2073276.2073282}.
However,
not all protocol actions that
are related to an attack are necessarily malicious. Consider, e.g., an
attack that relies on the public key retrieved from a key server ---
the sending of the public key is causally related to the violation,
but harmless in itself.
While causally related protocol actions can be a filter and a useful
forensic tool, they refer us back to the original question: What
constitutes malicious misbehavior?

Künnemann \etal~\cite{cryptoeprint:2018:127} answer these questions by considering \emph{the
fact} that a party deviated as a potential
cause~\cite{cryptoeprint:2018:127}. The main difference between their
framework and the present is that
they define a calculus which allows for individual parties to deviate. In
particular, deviating parties can communicate with each other and make their
future behavior dependent on such signals. This allows to encode
`artificial' causal dependencies in their behavior, e.g., a party
$A$, instead of mounting an attack, waits for an arbitrary message
from a second party $B$ before doing so. If one only observes $A$'s
attack, the involvement of $B$ is as plausible as the idea of $A$
acting on its own. This kind of `provocation' can occur whenever $B$
can secretly communicate with $A$. As out-of-band channels cannot be
excluded in practice, this shows that accountability is
impossible. In the single-adversary setting, provocation
cannot be encoded: all deviating parties are represented by the
adversary; neither $B$ sending the provocation message to $A$, nor $A$
conditioning the attack on the arrival of this message can be
expressed. As the case studies show, the provocation problem vanishes.
%
%
This implies that the single-adversary setting comes at a loss of
generality --- although it is a commonly assumed worst-case
assumption in security. They propose `knowledge-optimal' attacks,
where only secret information is exchanged, which we \emph{conjecture} to
be equivalent to the single-adversary setting following~\cite[Lemma
3]{cryptoeprint:2018:127}.
%
%
\begin{full}
Furthermore, they consider only 
a single counterfactual relation, which is
comparable with $r_c$ and fixes the
control-flow of all parties, except for those reverted to their
normative behavior, i.e., $\corrupted(t)\setminus \corrupted(t')$.
We chose to only fix the behavior of trusted parties,
as corrupted parties are controlled by the adversary, who, in our
calculus, does not have control-flow in the sense honest parties have.
\end{full}
\begin{full}
However, if one arrives at expressing this condition, e.g., via
messages transmitted, Lemma~\ref{thm:soundness} can be exploited to
automate most of the verification.
\end{full}
\begin{full}
Furthermore, the requirement that the violation occurs even if some
superset of $\B$ is corrupted is removed, as it follows immediately.
This is due to the fact that, again, causal dependencies between
deviating parties disappear in our corruption model.
\end{full}

Independently,
Bruni, Giustolisi and Sch{\"{u}}rmann,
propose a definition of accountability that is based on
parties choosing to deviate~\cite{DBLP:conf/isw/BruniGS17}.
In contrast to our verdict function, they consider per-party tests,
e.g., `Is party $A$ to be held accountable in this run'. 
Any set of these
tests can be interpreted as a verdict function that outputs 
the set of singleton sets for which the test gave true.
In fact, this is how we constructed the verdict function in the Certificate
Transparency case study from their tests. 
The downside to this approach is that joint accountability
is not expressible, excluding, e.g., the case studies in
Section~\ref{sec:casesrc}.
Nevertheless, it is instructive to compare their definition
to our verification conditions for $r_w$ (other relations are not
    covered).
The definition has four criteria, three of which have logically
    equivalent formulae in our set of verification conditions.
Their fourth criterion, however, is weaker (than $\Verifiability$), and our
criterion $\SufficiencyB$ is missing in their definition.
We argue that both conditions are strictly necessary: without
$\SufficiencyB$, a protocol that does not even permit a violation ($\varphi$ is always
    true), could fulfill all their criteria and still (unfairly!) blame
a party.
Furthermore, their counterpart to $\Verifiability$
allows violations to remain undetected under certain circumstances. We elaborate this
in 
    Appendix~\fvref{sec:giustolisi}.
This confirms our top-down approach: by building on 
accountability as a problem of causation, 
we were able to ground our verification conditions in
Def.~\ref{def:a-posteriori-verdict}.
We are able to specify
exactly what our verdict function is computing and can be sure to not
forget conditions.


\section{Conclusion and future work}\label{sec:concl}

We demonstrated the practicality of verifying accountability in
security protocols with a high degree of automation, 
and thus
call for the analysis of
existing and the invention of new protocols that provide for
accountability, instead of blindly trusting third
parties, e.g., voting protocols.
\begin{full}
While admittedly, there is no clear consensus on the choice of
counterfactuals, 
our method is adaptable to many relations,
and our case studies suggest that the one we chose for a closer look
works well.
Researchers in the AI and philosophy communities are currently
advancing our understanding of the matter. We stand to benefit from
their findings.
\end{full}

The present definitions apply to a wide-range of protocol
calculi. We implemented support for the SAPiC calculus, which allows for
a precise definition of control-flow. However, the definition we
provide cannot express computational or statistical indistinguishability.
For this reason, we had to simplify Kroll's accountable
algorithms protocol, and ignore its ability to randomize the
computation using a verifiable random function, and thus implement,
e.g., accountable lotteries. 
A computational variant of our definition would be desirable, but it
poses some technical challenges. In particular, counterfactual
adversaries shall not depend on the randomness actually used, which
prohibits a straight-forward translation and thus renders the
formulation of a generalized accountability game an interesting
question for future research.


\subsubsection*{Acknowledgements}
We thank Deepak Garg and Rati Devidze for discussion and advice in the early
stages of this work, and Hizbullah Abdul Aziz Jabbar for help on the
accountable algorithms protocol.
This work has been partially funded
by the German Research Foundation
(DFG) via the collaborative research center ``Methods and Tools for
Understanding and Controlling Privacy'' (SFB 1223),
and
via the ERC Synergy Grant IMPACT (Grant agreement ID: 610150).
The second author holds the Ministry of National Education Scholarship
of the Turkish Republic.

\printbibliography
\appendix
\subsection{Operational semantics}\label{app:semantics}

\subsubsection*{Frames and deduction} Before giving the formal semantics
of SAPiC, we introduce the notions of frame and deduction. A
\emph{frame} consists of a set of fresh names $\tilde n$ and a
substitution $\sigma$, and is written $\nu \tilde
n. \sigma$. Intuitively, a frame represents the sequence of messages
that have been observed by an adversary during a protocol execution
and secrets $\tilde n$ generated by the protocol, a priori unknown to
the adversary. Deduction models the capacity of the adversary to
compute new messages from the observed ones.

\begin{definition}[Deduction]
  \label{def:apip-ded}
  We define the deduction relation $\nu \tilde n . \Subst \vdash t$ as
  the smallest relation between frames and terms defined by the
  deduction rules in \autoref{fig:deduc}.
\end{definition}

\begin{figure*}[ht]
  $$
  \begin{array}{c @{\qquad}c}
  \infer[ \textsc{Dname}]{ \nu \tilde n . \Subst \vdash a }{ a\in\FN\cup\PN & a\notin\tilde n }
  &
  \infer[ \textsc{DEq}]{\nu \tilde n . \Subst \vdash t' }{\nu \tilde n . \Subst \vdash t & t=_E t'}
  \\[2mm]
  \infer[ \textsc{DFrame}]{\nu \tilde n . \Subst \vdash x\Subst}{x \in
    \Dom(\Subst) }
  & 
  \infer[ \textsc{DAppl}]{\nu \tilde n . \Subst \vdash f(t_1,\ldots,t_n) }
  {\nu \tilde n . \Subst \vdash t_1  \cdots 
    \nu \tilde n . \Subst \vdash t_n & f\in\Sigma^k\setminus
  \Sigma^k_\mathit{priv} }
  \end{array}
  $$
  \caption{Deduction rules.}
  \label{fig:deduc}
\end{figure*}

\begin{full}
\begin{figure*}
$
\infer
 {\nu \tilde n . \sigma \vdash k_2 }{  
\infer
{
\nu\tilde n.\sigma\vdash  \mathit{sdec}(c,k_1)
}
{
\infer
{\nu\tilde n.\sigma\vdash c
}
{ x_1\in\dom(\sigma)}
&
\infer
{\nu\tilde n.\sigma\vdash  k_1}
{ x_2\in\dom(\sigma)}
}
&
\mathit{sdec}(c,k_1) =_\ET k_2
 }
 $

\caption{Proof tree witnessing that $\nu \tilde n . \sigma \vdash k_2$, where $c=\mathit{senc}(k_2,k_1)$}
\label{fig:examplededuction}
\end{figure*}

\begin{example}
  If one key is used to encrypt a second key, then, if the intruder learns
  the first key, he can deduce the second. For $\tilde n=k_1,k_2$ and
  $\sigma=\set{^{\mathit{senc}(k_2,k_1)}/_{x_1}, ^{k_1}/_{x_2}}$,
  $\nu\tilde n.\sigma\vdash k_2$, as witnessed 
  by the proof tree given in Figure~\ref{fig:examplededuction}.
\end{example}
\end{full}

\subsubsection*{Operational semantics} We can now define the operational
semantics of our calculus. The semantics is defined by a labelled transition
relation between process configurations. A \emph{process
  configuration} is a 
  5-tuple $(\Names, \StoreA, \Processes, \Subst, \ActiveLocks)$
where 
\begin{itemize}
\item $\Names \subseteq \FN$ is the set of fresh names generated by the processes;
\item $\StoreA: \Mess_\Sign \to \Mess_\Sign$ is a partial function
  modeling the store;
\item $\Processes$ is a multiset of ground processes representing the
  processes executed in parallel;
\item $\Subst$ is a ground substitution modeling the messages output
  to the environment;
\item $\ActiveLocks \subseteq  \Mess_\Sign$ is the set of currently active locks
\end{itemize}

\begin{figure*}
  \input{fig-semantics-standard}
  \caption{Operational semantics.}
  \label{fig:operationalsemantics}
\end{figure*}

The transition relation is defined by the rules in
Figure~\ref{fig:operationalsemantics}. Transitions are labelled by
sets of ground facts. For readability, we omit empty sets and 
brackets around singletons, i.e., we write $\to$ for
$\stackrel{\emptyset}\longrightarrow$ and $\stackrel{f}\longrightarrow$ for $\stackrel{\set f}\longrightarrow$. We write $\to^*$
for the reflexive, transitive closure of $\to$ (the transitions that
are labelled by the empty sets) and write $\stackrel{f}
\Rightarrow$ for $\to^* \stackrel{f}\rightarrow \to^*$. 
We can now define the set of traces, i.e., possible executions that a process admits. 

\begin{definition}[Traces of $P$]\label{def:tracesppi}
    Given a ground process $P$,
    we define the \emph{traces of $P$} as
\begin{multline*}
    \traces(P) =  \left\{( F_1,\ldots,F_n ) \mid 
       \vphantom{\pitrans{F_1}_*} \right. 
    \\
      c_0
  \pitrans{F_1}_* 
  \ldots
  \pitrans{F_n}_*
  c_n  
    \left. \vphantom{\pitrans{F_1}_*} 
    \right\}, \text{ where }
\end{multline*}
$c_0 = (\emptyset,\emptyset, \set{P}, \emptyset, \emptyset, \emptyset)$,
A \emph{progressing trace} additionally fulfils the condition that all
    processes in $\Processes_n$ are blocking~\cite[Def.\ 2]{BaDrKr-2016-liveness}.
\end{definition}

If we are interested in liveness properties, we will only consider the set of \emph{progressing} traces, i.e.,
traces that end with a \emph{final} state. Intuitively, a state is final if all messages on resilient channels have been delivered and the process is \emph{blocking}~\cite[Def.~2,3]{BaDrKr-2016-liveness}.


\subsection{Security properties}\label{app:logic}

In the \Tamarin tool~\cite{SMCB-csf12}, security properties are
described in an expressive two-sorted first-order logic. The sort
\temp is used for time points, $\Vars_\temp$ are the temporal
variables.

\begin{definition}[Trace formulae] A trace atom is either false
  $\bot$, a term equality $t_1 \approx t_2$, a timepoint ordering $i
  \lessdot j$, a timepoint equality $i \doteq j$, or an action $F@i$
  for a fact $F\in\calF$ and a timepoint $i$. A trace formula is a
  first-order formula over trace atoms. 
\end{definition}
( If clear from context, we use
$t_1=t_2$ instead of $t_1\approx t_2$,
$i<j$ instead of $i \lessdot j$,
and
$i=j$ instead of$i \doteq j$.)


To define the semantics, let each sort $s$ have a domain $\Dom(s)$.
$\Dom(\tempsort)=\mathcal{Q}$,
$\Dom(\msgsort)=\Mess$,
$\Dom(\freshsort)=\FN$, and
$\Dom(\pubsort)=\PN$.
A function $\theta:\Vars \to \calM \cup \mathcal{Q}$ is a valuation if it
respects sorts, i.e., $\theta(\Vars_s) \subset \Dom(s)$ for all sorts
$s$. If $t$ is a term, $t\theta$ is the application of the homomorphic
extension of $\theta$ to $t$.

\begin{definition}[Satisfaction relation] 
  The satisfaction relation $(\tr,\theta) \vDash \varphi$ between
    a trace \tr, a valuation $\theta$, and a trace formula $\varphi$
    is defined as follows:
  \begin{align*}
      (\tr,\theta) &\vDash \bot && \ \text{never}\\
    (\tr,\theta) &\vDash F@i & &\iff \theta(i) \in \mathit{idx}(tr) \land  F\theta  \in_E \tr_{\theta(i)}\\
    (\tr,\theta) &\vDash i\lessdot j & &\iff   \theta(i)
      < \theta(j)\allowbreak\\
    (\tr,\theta) &\vDash i\doteq j & &\iff  \theta(i) = \theta(j)\\
    (\tr,\theta) &\vDash t_1 \approx t_2 & &\iff  t_1\theta =_E t_2\theta\allowbreak\\
    (\tr,\theta) &\vDash \neg \varphi & &\iff  \text{not } (\tr,\theta)\vDash \varphi\\
    (\tr,\theta) &\vDash \varphi_1 \wedge \varphi_2 & &\iff (\tr,\theta)\vDash \varphi_1 \text{ and } (\tr,\theta)\vDash \varphi_2\\
    (\tr,\theta) &\vDash \exists x:s .\varphi & &\iff  \text{there is } u\in \Dom(s) \\
      &&& \phantom{\iff~} \text{such that }(\tr,\theta[x\mapsto u])\vDash \varphi.
  \end{align*}
\end{definition}
For readability, we define 
$t_1 \gtrdot t_2$ as $\neg (t_1 \lessdot t_2 \vee t_1
\doteq t_2 )$ and  ($\ledot, \not \doteq, \gedot$) 
as expected. We also use classical notational shortcuts such as
$t_1 \lessdot t_2 \lessdot t_3$ for $t_1 \lessdot
t_2 \wedge t_2 \lessdot t_3$ and $\forall i \leq j.\ \varphi$ for
$\forall i.\ i \leq j \to \varphi$.
When $\varphi$ is a ground formula we sometimes simply write $\tr \vDash
\varphi$ as the satisfaction of $\varphi$ is independent of the valuation.
\begin{definition}[Validity, satisfiability]
Let $\mathit{Tr} \subseteq (\Processes(\GroundFacts))^*$ be a set of
traces. 
A trace formula $\varphi$ is said to be \emph{valid} for $\mathit{Tr}$ (written
$\mathit{Tr}  \vDash^\forall \varphi$) if for any trace $\tr \in
\mathit{Tr}$ and any valuation $\theta$ we have that  $(\tr , \theta) \vDash
\varphi$.

A trace formula $\varphi$ is said to be \emph{satisfiable} for $\mathit{Tr}$, written
$\mathit{Tr}  \vDash^\exists \varphi$, if there exist a trace $\tr \in
\mathit{Tr}$ and a valuation $\theta$ such that  $(\tr , \theta) \vDash
\varphi$.

\end{definition}

Note that $\mathit{Tr} \vDash^\forall \varphi$ iff $\mathit{Tr}
\not\vDash^\exists \neg \varphi$. 
Given a multiset rewriting system
$R$ we say that $\varphi$ is valid,  written $R \vDash^\forall
\varphi$, if $\tracesmsr(R) \vDash^\forall \varphi$.  We say that
$\varphi$ is satisfied in $R$, written $R \vDash^\exists \varphi$, if
$\tracesmsr(R) \vDash^\exists \varphi$.  Similarly, given a ground
process $P$ we say that $\varphi$ is valid, written $P \vDash^\forall
\varphi$, if $\traces(P) \vDash^\forall \varphi$, and that $\varphi$
is satisfied in $P$, written $P \vDash^\exists \varphi$, if
$\traces(P) \vDash^\exists \varphi$.

\begin{full}
\subsection{Proofs for
Section~\ref{sec:cond-coarse}}\label{sec:proof-cond-coarse}

\begin{theorem}[soundness of verification conditions]\label{thm:soundness-coarse}
    For any protocol $P$,
    predicate on traces $\varphi$
    and verdict function $\verdict: \traces(P)\to2^{2^\Agents}$
    of the form described in
    Definition~\ref{def:verification-conditions-coarse},
    if the conditions $\gamma_{\varphi,\verdict}$ hold,
    then $\verdict$ provides $P$ 
    with accountability for $\varphi$ for $r=r_w$.
\end{theorem}
\begin{proof}
Note that while our implementation checks for the conditions (XH) and
    (XC) to make sure that $V_1$ to $V_n$ and $\omega_1$ to
    $\omega_n$ indeed define a function (XC) that is total on the
    traces (XH), this is already  required by definition of a verdict
    function. 

    We need to show that for all $t\in\traces(P)$ 
    the actual verdict
    coincides with the a~posteriori verdict.
    We fix $t$ and $i$ such that $t\vDash \omega_i$.

    If $V_i=\emptyset$, then,
    by $\Verifiability$,
    $\varphi(t)$.
    By definition of $r_w$ and the a~posteriori verdict,
    $\varphi(t)$
    implies
    $\apv_{P,\varphi,r_w}(t)=\emptyset=V_i$.

    If $V_i\neq\emptyset$,
    from $\SufficiencyB$
    we have that for each $\B\in V_i$,
    there exists $t'$ such that
    $\corrupted(t')=\B\neq \emptyset$
    and  
    $t'\vDash \neg \varphi$.
    From $\Uniqueness$, we get that 
    $\corrupted(t') 
    = \B
    \subseteq \cup_{\B'\in V_i} \B'
    \subseteq \corrupted(t)$
    and thus
    $r_w(t,t')$.
    From $\Minimality$, we get that each such $t'$ is
    minimal.
    Finally, from $\Verifiability$, we have 
     $t\vDash \neg \varphi$
    and hence $\B\in\apv_{P,\varphi,r_w}(t)$.

    It is left to show that all $\B\in\apv_{P,\varphi,r_w}(t)$ are in
    $V_i\neq\emptyset$.
    Fix $t$, $i$ and $\B$.
    By definition of the a~posteriori verdict,
    there is $t'$ such that $r_w(t,t')$, $\neg\varphi(t')$ and
    $\corrupted(t')=\B$.
    By $\Exhaustiveness$, there is some $j$ s.t.\ $\omega_j(t')$.
    We show by case distinction, that $V_j$ is singleton.
    Case (a): $V_j=\emptyset$. Contradiction to $\neg\varphi(t')$ and
    $\Verifiability$.
    Case (b): $|V_j|\ge 2$.
    By $\Uniqueness[j]$, 
    \begin{equation}
        \bigcup_{\B'\in V_j} \subseteq \corrupted(t').\label{eq:1}
    \end{equation}
    By $\SufficiencyB[j]$, for any such $\B'\in V_j$, we have
    $t''$ with $\neg\varphi(t'')$ and $\corrupted(t'')\in V_j$.
    Furthermore, 
    \begin{equation}
        \corrupted(t'')\subsetneq \bigcup_{\B''\in
        V_j},\label{eq:2}
    \end{equation}
    as otherwise, there would be $\B''\in V_j$ that is different from
    $\B'$ but a subset of it, in which case $\SufficiencyB[j]$ and $\Minimality[j]$ would be
    in contradiction.
    From \eqref{eq:1} and \eqref{eq:2},  we conclude that
    $\corrupted(t'')\subsetneq \corrupted(t')$, contradicting the
    minimality of $t'$.

    We are thus left with the last case. Case (c): $V_j=\set{\B'}$.
    By $\Uniqueness[j]$, $\B'\subseteq \corrupted(t')=\B$.
    If $\B'\subsetneq \B$, then by $\SufficiencyB[j]$, there is $t''$
    such that $\neg\varphi(t'')$, $r_w(t',t'')$ and
    $\corrupted(t'')=\B'$. By transitivity of $r_w$,
    $r_w(t,t'')$, and thus this case would contradict minimality of $t'$.
    Hence $\B'=\B$. In this case, $\Completeness$ guarantees that
    $\B=\B'\in V_i$.
    Therefore,
    $\apv_{P,\varphi,r_w}(t) 
    = \set{ \B \mid \B\in V_i}
    = V_i = \verdict(t)$.
\end{proof}

If, for each case $\omega_i$, there is $t\in\traces(P)$ such that $t
\vDash \omega_i$, then these conditions are complete.

\begin{theorem}[completeness of verification conditions]
    For any protocol $P$,
    predicate on traces $\varphi$
    and verdict function $\verdict: \traces(P)\to2^{2^\Agents}$
    of the form described in
    Definition~\ref{def:verification-conditions-coarse},
    if $\verdict$ provides $P$ 
    with accountability for $\varphi$ for $r=r_w$,
    and for every $\omega_i$, there is $t\in\traces(P)$ s.t.\
    $\omega_i(t)$,
    then
    the conditions $\gamma_{\varphi,\verdict}$ hold,
\end{theorem}
\begin{proof}
    Fix $\verdict$, $P$ and $\varphi$.
    We argue the conditions one by one, referring only back to
    conditions that have already been established.
    $\Exhaustiveness$ and $\Exclusiveness$ hold because $\verdict$ is
    a total function.
    $\SufficiencyB$ follows directly from
    Definitions~\ref{def:a-posteriori-verdict}
    and~\ref{def:accountability}.
    From Definition~\ref{def:a-posteriori-verdict},
    we can follow that
    $\apv_{P,\varphi,r_w}(t)=\emptyset \iff t \vDash \neg \varphi$.
    Thus $\Verifiability$ holds from the assumption that $\verdict$
    provides accountability.
    $\Minimality$: 
    By contradiction, if $\B\in V_i$ contradicts $\Minimality$,
    then, by $\SufficiencyB$, there is $t_{\B}$ such that 
    $\corrupted(t_{\B})=\B$ and $t_{\B}\vDash \neg \varphi$,
    but for $\B'\subsetneq \B$, there is $t_{\B'}$ s.t.\ $r_w(t_\B,t_{\B'})$
    and $t\vDash \neg\varphi$.
    By accountability and by the assumption that there is some
    $t\in\traces(P)$
    s.t.\ $t\vDash \omega_i$, we have that $r_w(t,t_{\B})$ and, by
    transitivity, $r_w(t,t_{\B'})$. Hence $\B$ not minimal in
    $\apv_{P,\varphi,r_w}(t)$, contradicting accountability.
    $\Uniqueness$: Fix any $t$ and $i$ s.t.\ $t\vDash \omega_i$, 
    By accountability, we have
    $\verdict(t)=\apv_{P,\varphi,r_w}(t)$,
    and thus, by definition of the a~posteriori verdict and $r_w$,
    for all $\B'\in V_i$, 
    $\B'\subseteq \corrupted(t)$.
    Therefore, 
    $\cup_{\B'\in V_i} \B' \subseteq \corrupted(t)$
    for any $\omega_i$ for which such a $t$ exists.
    $\Completeness$: For contradiction, let there be $j$, $\B$, $i$ s.t.\ $V_j=\set{\B}$,
    $\B\not\in V_i$
    but
    $\B\subseteq \cup_{\B'\in V_i} \B'$.
    By $\Uniqueness$, 
    $\cup_{\B'\in V_i} \B' \subseteq \corrupted(t)$,
    and thus $r_w(t,t_\B)$
    for $t$ s.t. $t\vDash \omega_i$ (which exists by assumption)
    and $t_\B$ s.t.\ $\corrupted(t_\B)=\B$ and $t_\B\vDash
    \neg\varphi$ (which exists by $\Sufficiency[j]$).
    By $\Minimality$, $\B$ is minimal, hence $\B\in V_i$ by
    accountability, contradicting the assumption.
\end{proof}

\subsection{Proofs for Section~\ref{sec:ver}}\label{sec:proof-cond-cases}

All but the last two conditions can be readily checked using 
SAPiC/Tamarin.
We will implement a check for the case where $r=r_c$, i.e., where
the relation 
captures the requirement that counterfactuals expose the same
control-flow for trusted parties, similar (but not equal) to the
relation used in~\cite{cryptoeprint:2018:127}.
But first, we prove that these conditions sufficient.
\begin{theorem}[soundness]\label{thm:soundness}
    For any protocol $P$,
    predicate on traces $\varphi$
    and verdict function $\verdict: \traces(P)\to2^{2^\Agents}$
    of the form described in
    Definition~\ref{def:ver-conditions},
    if the conditions $\nu_{\varphi,\verdict}$ hold,
    then $\verdict$ provides $P$ 
    with accountability for $\varphi$.
\end{theorem}
\begin{proof}
Note that while our implementation checks for the conditions
    $\Exhaustiveness$ and
    $\Exclusiveness$ to make sure that $V_1$ to $V_n$ and $\omega_1$ to
    $\omega_n$ indeed define a function ($\Exclusiveness$) that is total on the
    traces ($\Exhaustiveness$), this is already  required by definition of a verdict
    function. 

    We need to show that for all $t\in\traces(P)$ 
    the actual verdict
    coincides with the a~posteriori verdict.
    We fix $t$ and $i$ such that $\omega_i(t)$.

    If $V_i=\emptyset$, then,
    by $\Verifiability$,
    $\varphi(t)$.
    By definition of $r$ and the a~posteriori verdict,
    $\varphi(t)$
    implies
    $\apv_{P,\varphi,r}(t)=\emptyset=V_i$.

    If $V_i=\set{\B}$, then,
    by $\Verifiability$,
    $\neg\varphi(t)$.
    From $\SufficiencySingleton$,
    there exists $t'$ s.t.\ 
    $\corrupted(t')=\B$
    and  $t'\vDash \neg \varphi$.
    $\RelationReflexive$ and $\RelationLifting[i]$ imply $r(t,t')$.
    From $\MinimalitySingleton$ we get that, 
    for any $t^*$ such that $\corrupted(t^*)$ is a strict subset of
    $\B$ and
    $j$ such that $\omega_j(t^*)$ (by $\Exclusiveness$ and
    $\Exhaustiveness$, there is exactly one),
    $j\neq i$. By $\RelationReflexive$, 
    this means that $\neg
    R_{i,j}$.
    By $\RelationLifting$, this implies $\neg r(t,t^*)$ or $V_j
    = \emptyset$. From $\Verifiability[j]$, $V_j=\emptyset$
    implies $\varphi(t^*)$, hence in both cases
    $\corrupted(t^*)$ is not in the a~posteriori verdict.
    Hence $t'$ is minimal w.r.t.\ $\corrupted(t')$.
    From $\RelationReflexive$ and $\RelationLifting$,
    we can conclude that any other $t''$ 
    such that $r(t,t'')$ and $\neg\varphi(i)$
    has $\omega_i(t'')$.
    From
    $\UniquenessSingleton$,
    we follow that
    $\corrupted(t'')=\B$, or is not minimal.
    Thus
    $\apv_{P,\varphi,r}(t) = \set{ \B }
    = V_i = \verdict(t)$.

    If $V_i = \set{\B_1,\ldots,\B_n}$, $n\ge 2$, then,
    by $\Verifiability$,
    $\neg\varphi(t)$.
    From $\SufficiencyR$
    we have that 
    for each $\B\in V_i$,
    there is a $j$ such that $R_{i,j}$ and $V_j = \set{\B}$.
    For each $\B\in V_i$,
    from $\SufficiencySingleton$,
    there exists $t'$ such that
    $\corrupted(t')=\B$,
    $t'\vDash \neg \varphi$.
    As $\neg\varphi(t)$ and
    $\neg\varphi(t')$, we can conclude from
    $\RelationLifting$
    that
    $r(t,t')$.
    From $\MinimalitySingleton$ we get that,
    for any $t^*$ such that $\corrupted(t^*)$ is a strict subset of
    $\B$, and any
    $k$ such that $\omega_k(t^*)$ (by $\Exclusiveness$ and
    $\Exhaustiveness$, there is exactly one),
    $\neg R_{j,k}$. 
    From $\MinimalityR$, we obtain $\neg R_{i,k}$.
    By $\RelationLifting$, this implies $\neg r(t,t^*)$ or $V_k
    = \emptyset$. From $\Verifiability[k]$, $V_k=\emptyset$
    implies $\varphi(t^*)$, hence 
    in both cases
    $\corrupted(t^*)$ is not in the a~posteriori verdict.
    Hence each such $t'$ is minimal w.r.t.\ $\corrupted(t')$.
    From $\RelationLifting$,
    we can conclude that any other $t''$ 
    such that $r(t,t'')$ and $\neg\varphi(t'')$
    has $\omega_j(t'')$
    for some $j$ with $R_{i,j}$.
    From $\CompletenessR$ and 
    $\UniquenessSingleton[j]$,
    we follow that
    $\corrupted(t'')=\B$ for some $\B\in V_i$, or is not minimal.
    Thus
    $\apv_{P,\varphi,r}(t) = \set{ \B \mid \B\in V_i}
    = V_i = \verdict(t)$.
\end{proof}

As it turns out, the conditions are also complete. Practically, this
means that a counter-example to any lemma generated from the
conditions is also an attack against accountability. 

\begin{theorem}[completeness]\label{thm:completeness}
    For any protocol $P$,
    predicate on traces $\varphi$
    and verdict function $\verdict: \traces(P)\to2^{2^\Agents}$
    of the form described in
    Definition~\ref{def:ver-conditions},
    if 
     $\verdict$ provides $P$ 
    with accountability for $\varphi$,
    then  the conditions
    $\nu_{\varphi,\verdict}$ hold (for suitably chosen $R$ and
    $\omega_i$).
\end{theorem}
\begin{proof}
    Assuming that $\verdict$ provides $P$ with accountability for $\varphi$,
    we argue the conditions one by one, referring only back to
    conditions that have already been established.
    $\Exhaustiveness$ and $\Exclusiveness$ hold because $\verdict$ is
    a total function.
    As $r$ is 
    reflexive
    we can define equivalence classes
    with $t = t' \iff r(t,t') \land r(t',t)$.
    As $r$ is transitive, either all members of an equivalence class
    are related to all members of another equivalence class, or none
    is. 
    We can thus split the cases according to these equivalence
    classes and define $R_{i,j}$ iff
    $i$ and $j$ result from the same case and the members in the
    equivalence classes that split this case are related.
    W.l.o.g.\ we chose the case distinction so that every
    cases covers a trace, i.e., 
    for every $\omega_i$, there is
    $t\in\traces(P)$ such that $\omega_i(t)$.
    Then
    $\RelationLifting$ 
    and $\RelationReflexive$ hold.
    $\SufficiencySingleton$ follows directly from
    Definitions~\ref{def:a-posteriori-verdict}
    and~\ref{def:accountability}.
    $\SufficiencyR$: If $\B$ in $\apv_{P,\varphi,r}(t)$, then there is
    a minimal (w.r.t.\ $\corrupted(\cdot)$) $t'$ with $r(t,t')$, $\corrupted(t')=\B$, and
    $t'\vDash \neg \varphi$.
    By $\Exhaustiveness$, there is a $j$ such that $\omega_j(t')$,
    and from $r(t,t')$ and $\omega_i(t)$, we follow by
    $\RelationLifting$ that $R_{i,j}$.
    Therefore $V_j$ is singleton: if it were empty, $t'\vDash \neg\varphi$
    would not hold, if $|V_j|\ge 2$, then 
    there are two distinct $t''$ such that $r(t',t'')$, $\neg\varphi(t')$ and
    $\corrupted(t'')\subset\B$, hence at least one of them has
    $\corrupted(t'')\subsetneq\B$.
    By transitivity of $r$, this would contradict the minimality of
    $t'$.
    Thus there is a $j$ for every $\B\in V_i$ s.t.\ $V_j=\set{\B}$ with
    $R_{i,j}$.
    From Definition~\ref{def:a-posteriori-verdict},
    we can follow that
    $\apv_{P,\varphi,r}(t)=\emptyset \iff t \vDash \neg \varphi$.
    Thus $\Verifiability$ holds from the assumption that $\verdict$
    provides accountability.
    $\MinimalitySingleton$: By contradiction, if $\apv_{P,\varphi,r}(t)=\set{\B}$,
    but there was $t'$ such $\omega_j$, $R_{i,j}$ and
    $\corrupted(t')=\B'$ a strict subset of $\B$, 
    then $r(t,t')$ by $\RelationLifting$
    ($V_j\neq \emptyset$, as otherwise there could not be a strict
    subset $\B'$), 
    and thus by $\Verifiability$,
    $t''\vDash \neg \varphi$, contradicting minimality.
    $\MinimalityR:$ This would directly contradict minimality in
    Definition~\ref{def:a-posteriori-verdict}.
    $\UniquenessSingleton$: If there is $t$ with
    $\B  \not\subseteq  \corrupted(t)$, but
    $\apv_{P,\varphi,r}(t)=\set{\B}$,
    then by
    $\Verifiability$,
    $t \vDash \neg\varphi$.
    Let $t'$ the minimal trace (set inclusion order w.r.t.\
    $\corrupted(t')$), such that this holds. Then 
    $\corrupted(t')$ is either a strict subset of
    $\corrupted(t)$ (contradicting minimality),
    or an incomparable set (contradicting the assumption that
    $\apv_{P,\varphi,r}$ is singleton, and thus $V_i$ is singleton).
    %
    Finally, $\CompletenessR$:
    Let $t$ such that $\verdict(t)=\apv_{P,\varphi,r}(t)$
    and $\omega_i(t)$ with $|V_i|\ge 2$.
    Again, proof by contradiction: Assume $j$ with $R_{i,j}$, 
    $V_j=\set{\B}$
    but
    neither
    $\B\in V_i$ nor a superset of some $\B'\in V_i$.
    Due to 
    $\SufficiencySingleton[j]$
    and
    $\RelationLifting$, 
    there is $t'$ such that $r(t,t')$, $\neg\varphi(t')$ and
    $\corrupted(t')=\B$.

    We will now repeat part of the argument from the proof to
    Theorem~\ref{thm:soundness} to show that $\B$ is minimal and thus
    ought to be in $\apv_{P,\varphi,r}(t)$.
    From $\MinimalitySingleton[j]$ we get that, 
    for any $t^*$ such that $\corrupted(t^*)$ is a strict subset of
    $\B$ and
    $k$ such that $\omega_k(t^*)$ (by $\Exclusiveness$ and
    $\Exhaustiveness$, there is exactly one ),
    $k\neq j$. By $\RelationReflexive$, 
    this means that $\neg
    R_{j,k}$.
    From $\MinimalityR$ and the fact that $\B$ is not a superset of any
    $\B'\in V_i$, we obtain $\neg R_{i,k}$.
    By $\RelationLifting$, this implies $\neg r(t,t^*)$ or $V_k
    = \emptyset$. From $\Verifiability[k]$, $V_k=\emptyset$
    implies $\varphi(t^*)$, hence 
    in both cases
    $\corrupted(t^*)$ is not in the a~posteriori verdict.
    Hence each such $t'$ is minimal, and thus $\B\in\apv_{P,\varphi,r}(t)$,
    but $\B\not\in\verdict(t)$.
\end{proof}

\subsection{Sequential self-composition}\label{sec:ctl-full}

We solve this problem by combining $t$ and $t'$ in a single trace,
which is their concatenation.
If for all occurrences of $\Control(p)$ in the first part,
and for all occurrences of $\Control(p')$ in the second part $p=p'$,
then the $t$ and $t'$ corresponding to these two parts are in the
relation.

Observe, however, that $P;P$ is not necessarily a syntactically correct
process. To define sequential self-composition, we modified SAPiC's
translation from processes to multiset rewriting rules and SAPiC's 
translation from (SAPiC) lemmas to (Tamarin) lemmas.
We describe the former in terms of the difference to the original
SAPiC translation, to avoid a review of the complete translation:
\begin{enumerate}
    \item SAPiC annotates the initial multiset rewrite rule (msr) 
        with a fact $\Init()$ and restricts traces to those
        containing at most one occurrence of this fact. We modified
        this fact to contain a fresh execution identifier $\eid$ and
        allow multiple instances of this rule with distinct identifiers.
        \label{it:init}
    \item We added exactly one additional msr, which permits 
        emitting a fact $\Stop(\eid)$, but only after the
        aforementioned initial rule was instantiated with the same
        $\eid$.\label{it:stop}
    \item Wherever any fact (event) is emitted, we also
        emit an $\Event(\eid)$-fact at the same time-point,
        where $\eid$ is the execution identifier chosen at
        initialisation.\label{it:trans-event}
\end{enumerate}
We can now talk about two executions (distinguished by $\eid$)
within one trace, but they might be interleaved.
Within Tamarin, we can `restrict' the set of traces
by axiomatic statements in the same logic used to define security
properties. Proving $\varphi$ with restriction $\alpha$ essentially
boils down to proving $\alpha \implies \varphi$.
We use this mechanism to impose the following restrictions in addition
    to those imposed by SAPiC's translation. (For brevity, all free
    variables are universally quantified.)

\begin{itemize}
    \item Every execution is terminated:
        $\Init(\eid)@i \implies \exists k. \Stop(\eid)@k \land
        i < k$.\footnote{The notation $F@i$ means
    the fact $F$ was emitted (by an event) at time-point
    $i$. See Appendix~\ref{app:logic} for details.}
\item $\Event$-facts (and hence all other facts) are enclosed by
        $\Init$ and $\Stop$:
 $\Init(\eid)@i \land \Stop(\eid)@j \land \Event(\eid)@k \implies i < k \land k < j$.
    \item $\Init$ and $\Stop$ define disjoint intervals:
     $\Init(\eid_1)@i \land \Stop(\eid_1)@j \land \Init(\eid_2)@k \land \Stop(\eid_2)@l \implies
        (j < k \land j < l) \lor (l < i \land l < j) \lor (i=k \land j=l)$.
\end{itemize}

Let $\AssSeq$ be the conjunction of these formulae. 
We rewrite a formula $\varphi$
to apply to a given
execution identified by $\eid$ by substituting any atom of form 
$F@i$ by the conjunction $F@i \land \Event(\eid)@i$. We denote this
rewritten formula by $\varphi^\eid$.
%
%
%
We can now
express one direction of the co-implication $\RelationLifting$, i.e.,
for $V_i,V_j\neq \emptyset$ and $R_{i,j}$,
as follows:
\begin{align*}
    \AssSeq & \land 
    \forall i,j,\eid_1,\eid_2 . \Init(\eid_1)@i \land \Init(\eid_2)@j \\
    &\land \omega_i^{\eid_1} \land \omega_j^{\eid_2} \\
     \implies & 
     \forall p_1,t_1,p_2,t_2 .\, \Control(p_1)@t_1 \land \Event(\eid_1)@t_1 \\
    & \land \Control(p_2)@t_2 \land \Event(\eid_2)@t_2 \\ 
    & \implies p_1 =_\ET p_2.
\end{align*}
Intuitively, if the two executions identified by $\eid_1$ and $\eid_2$
are properly aligned (w.r.t.\ 
$\AssSeq$) and fall into cases $\omega_i$ and $\omega_j$,
then the control-flow chosen coincides.
The reverse implication, i.e., for
$V_i,V_j\neq \emptyset$ and $\neg R_{i,j}$,
is expressed using the same formula, 
but negating the top-level consequent ($\forall p_1,t_1,p_2,t_2. \cdots$).

The soundness and completeness of this approach is not difficult to
show (see Appendix~\ref{sec:ctl-proof}),
but depends on how exactly SAPiC translates protocols into multiset
rewrite rules and is thus rather technical.
We only sketch the
argument here: For soundness, the restrictions in
$\AssSeq$ only weakens the lemma. The rewriting of $\omega_i$ to
$\omega_i^\eid$ is correct by Item~\ref{it:trans-event} of the
translation. The consequent encodes the definition of $r_c$.
For completeness, observe that processes are ground, hence messages
emitted or compared to in the second run cannot contain names from the
first run -- unless private function symbols are used, in which case
completeness does not hold.
Barring private function symbols, these messages thus do not improve
the deductive capabilities of the adversary. This can be shown by the
induction over the four message deduction rules.
Hence, at each point, the second execution of the process
can reduce iff they could reduce in the first run. 


\subsubsection{Application to other relations.}

In some cases, the same method can even be used to capture the
relation $r_v$, which relates actual trace and counterfactual trace
that contain the same \emph{instance} of the violation. 
Consider the following example: Let $\varphi\defeq \forall
S,i. \mathit{Action}(S,O,a)@i \implies \exists j.
\mathit{Request}(S,O,a)@j$, i.e., for any action $a$ performed by
a subject $S$ on an object $O$, there was a previous request. Note
that $O$ and $a$ are free, i.e., we consider the object and the action
important, but not the subject, so that
$r_v$ consider violations the same if they pertain 
to the same object and the same action.

As any violation is characterized by an $\mathit{Action}$-event
without matching request, we can use the above mechanism to verify
$\RelationLifting$ for $r_v$. We annotate each event
$\mathit{Action}(S,O,a)$ with an additional event $\Control(\langle
O,a\rangle)$. This method applies if we can identify a set of
event-constructs in the processes that characterize all violations,
e.g., if the last reduction of any length-minimal violating execution
is a reduction of such an event from this set.

\subsubsection{Limitations}

Note that for a process where a trusted party is under replication, it
is possible that two different $\Control()$-events are emitted in the
same execution, and thus the corresponding process is in no
equivalence class w.r.t.\ $r_c$.
This affects only one of our case studies (centralized monitor),
however, instead of considering a possibly replicating series of
violations, we chose to identify the party causing the \emph{first}
violation.

We nevertheless discuss how to generalize $r_c$ and caveats in doing
so. Consider $r_c'(t,t')$ iff for all $i$, $p$ and $p'$,
$\Control(i,p)\in t$ and $\Control(i,p')\in t'$
implies $p=_\ET p'$.
Here, $\Control()$ events can be guaranteed to be different via $i$,
which can, e.g., contain  
a counter, or an incoming message (if these are unique).
Philosophically, however, it is not clear whether the counterfactual
control-flow needs to have the same order as the actual control-flow,
whether the order is irrelevant and $t'$ should just describe the
set (or a subset of) $\set{ p | t \vDash \exists i. \Control(i,p)}$.
Practically, $r_c'$ also requires the verdict function to aggregate the
verdict over all violations that occurred.
Technically, $r_c'$ requires the implementation of
first-order cases that can treat
several instances of $\RelationLifting$ in a single
SAPiC/Tamarin-lemma, as there can be an infinite number of
equivalence classes w.r.t.\ $r_c'$. 
These options of generalising $r_c$ need to be compared and evaluated,
which we consider out of the scope of this work.
Conceptually, however, most our results transfer to more refined
versions of $r_c$:
Theorem~\ref{thm:soundness}~and~\ref{thm:completeness} apply no
matter which relation is chosen (their proofs also apply to a countably
infinite number of cases), and the method for sequential
self-composition presented here applies as long as $r_c'$ can be
described as a property over two traces.

\end{full}
\begin{full}
\subsection{Proofs for Section~\ref{sec:ctl-flow}}\label{sec:ctl-proof}

We use the notation and translation as defined in~\cite{KK-jcs16}. The
modified translation for the analysis for liveness
properties~\cite{BaDrKr-2016-liveness} uses this proof largely in a black-box
manner. The proof we present here is described in terms of
modifications to the proof in~\cite{KK-jcs16},
the black-box proof in~\cite{BaDrKr-2016-liveness}
applies in a straight-forward manner.

Additionally, we introduce the following notation:
For $t=(t_1,\ldots,t_n)$ and $e\in\FN$, let
$t^e$ denote
$(t_1,\Event(e),\allowbreak,\ldots,t_n,\Event(e))$.

We modify each axiom in $\Ass$ to pertain to only one execution, i.e.,
$\Ass'\defeq{} \forall e. \Ass^e$ is considered. We modify $\filter$
accordingly. The set of reserved names (and thus $\hide$) is extended
with $\Stop$ and $\Event$. (As demonstrated~\cite{BaDrKr-2016-liveness}, the
axiom $\AssIn$ can be disregarded for a large class of properties,
which includes all properties we have encountered so far. As this
improves performance, our implementation also ignores this axiom.) 

\begin{theorem}[soundness of sequential self-composition]
    Given a well-formed ground process $P$ and
    $t_1,t_2\in\hide(\allowbreak\filter(\tracesmsr(\sem{P})))$
    and fresh names $e_1,e_2\in\FN$,
    $t_1^{e_1}\cdot t_2^{e_2} \in \hide(\filter(\tracesmsr(\sem{P}_a)))$.
\end{theorem}
\begin{proof}
    Observe that $t_1^{e_1}$, as the modification of the initial msr
    permits choosing $e_1$, which carries over using the $\state$
    variables. (This can be shown via induction, see
    \cite[Lemma~10]{KK-jcs16}.)
    As any fact (event) that is emitted is marked additionally with
    $\Event(v)$ for $v$ bound to the variable position of $e_1$ in the
    \state-fact, $\Event(v)$ follows any $t_1^i\in t_1$.

    The additional msr $[\Stop(e_1)]\msrewrite{\Stop(e_1)}[]$ permits
    emitting the action $\Stop(e_1)$ so that it appears at position
    $|t_1^{e_1}|+1$, as by definition of the translation, $\Init(e_1)$
    must have appeared before any other (protocol) action, and hence
    the fact $\Stop(e_1)$ is in the state after the transition
    producing the last element of $t_1^{e+1}$.
    This ensures the two first conditions of
    $\AssSeq$. With the same argument, the trace in $t_2^{e_2}$ can be
    reproduced, terminated with $\Stop(e_2)$.
    The argument is literally the same as for the soundness of the
    translation
    (\cite[Theorem~1]{KK-jcs16}, in particular \cite[Lemma~10]{KK-jcs16} for the
    core argument), however, the adversary has additional knowledge
    facts. 
    Hence we only need to modify the cases in \cite[ Lemma~10 ]{KK-jcs16}
    pertaining to message deduction.
    In these cases, we modify every reference to \cite[Lemma~8]{KK-jcs16} (the
    second part, more precisely) to 
    refer to 
    \cite[Lemma~7]{KK-jcs16} and
    Lemma~\ref{lem:strenghtening} below, showing that for
    $\sigma=\sigma_1\cap \sigma_2$, where $\sigma_1$ is the knowledge
    from the first run,
    and 
    $\tilde n = \tilde n_1 \cap \tilde n_2$,
    where
    $\tilde n_1$ are the nonces chosen in the first run,
    $\nu \tilde n , \sigma \vDash t$
    if
    $\nu \tilde n_2, \sigma_2 \vDash t$,
    as long as the nonces in $t$, in $\tilde n_1$
    and in $\sigma_1$ are disjunct.

    The resulting trace ensures the third condition, as well as (with
    the same argument as before) the other conditions of $\AssSeq$.
\end{proof}

For completeness, we need to introduce an additional caveat. The
adversary can use knowledge from previous protocol runs. As nonces are
fresh, this can only help if no private function symbols are used.
We therefore state the completeness result for the case where 
$\Sign_\mathit{priv}=\emptyset$.

Luckily, this additional knowledge is not useful, as the following
lemma shows for message deduction in SAPiC.
In combination with \cite[ Lemma~8 ]{KK-jcs16},
which relates message deduction in SAPiC to message deduction in
Tamarin, this can be used to argue about the modified translation.

\begin{lemma}\label{lem:strenghtening}
    Let $m$ be ground and
    $\sigma$ and $\sigma'$, as well as,
    $\mathfun{nonces}(\sigma)$ and $t$
    be
    disjunct. Then
    $\nu \tilde n. \sigma' \vdash m \implies
    \nu \tilde n. \sigma \cap \sigma' \vdash$. The opposite direction
    holds, if $t$ (or any $t' =_E t$) contains no private function symbols.
\end{lemma}
\begin{proof}
    Structural induction over the deviation of $m$. Case distinction
    over the four rules. \textsc{DName}, \textsc{DFrame} are
    trivial. \textsc{DAppl}: by induction hypothesis. \textsc{DEq}:
    $\ET$ is closed under renaming.
\end{proof}

\begin{theorem}[completeness of sequential self-composition]
    Let $\Sign_\mathit{priv}=\emptyset$.
    Given a well-formed ground process $P$,
    if $t\in\hide(\filter(\tracesmsr(\sem{P}_a)))$,
    then, 
    there exist
    $t_1,t_2\in\hide(\filter(\tracesmsr(\sem{P})))$,
    (with $t_2$ possibly empty)
    such that for some fresh names $e_1,e_2\in\FN$,
    $t_1^{e_1}\cdot t_2^{e_2} \cdot t' = t$.
\end{theorem}
\begin{proof}
    We can split $t$ into three parts, the first two terminated with
    the earliest action $\Stop(e_1)$ for some name $e_1$, or
    $\Stop(e_2)$ respectively.
    By $\AssSeq$, either $t$ only contains deduction actions
    $K(\cdot)$, or at least one action $\Stop(e_1)$.
    We apply \cite[Lemma~11]{KK-jcs16} to obtain a normalized execution
    that reproduces $t$. 
    In the first case, we
    set both the second and third part to the empty sequence.
    Due to the fact that $\Stop$ is now considered a reserved event, 
    we can literally apply \cite[Lemma~12]{KK-jcs16}.
    As the translation was modified to contain $\Event(e_1)$ in every
    step with an action, we have that $t^{e_1}=t$. 

    In the second case, there is an action $\Stop(e_1)$.
    With the previous argument, 
    we can set $t_1^{e_1}$ to the prefix of $t$ up to, and including,
    $\Stop(e_1)$.
    By $\AssSeq$ and the modification to the translation
    (Item~\ref{it:trans-event}),
    for any action that follows suit, the execution
    identifier needs to be different, i.e., it is annotated with
    $\Event(e_2)$ and $e_1 \neq e_2$. 
    By modification of the translation (Item~\ref{it:init}),
    there is a corresponding $\Init$ event, and, again by 
    $\AssSeq$, it is followed by a $\Stop(e_2)$ event.
    We set $t_2^{e_2}$ to be the remaining suffix of $t$ up to, and
    including, this $\Stop(e_2)$ event.
    We can apply \cite[Lemma~12]{KK-jcs16},
    but need to modify the
    cases in \cite[ Lemma~10 ]{KK-jcs16}
    pertaining to message deduction.
    In these cases, we modify every reference to \cite[Lemma~8]{KK-jcs16} (the
    second part, more precisely) to 
    refer to 
    \cite[Lemma~7]{KK-jcs16} and
    Lemma~\ref{lem:strenghtening}, showing that for
    $\sigma=\sigma_1\cap \sigma_2$, where $\sigma_1$ is the knowledge
    from the first run,
    and 
    $\tilde n = \tilde n_1 \cap \tilde n_2$,
    where
    $\tilde n_1$ are the nonces chosen in the first run,
    $\nu \tilde n , \sigma \vDash t$
    iff
    $\nu \tilde n_2, \sigma_2 \vDash t$,
    as long as the nonces in $t$, in $\tilde n_1$
    and in $\sigma_1$ are disjunct and since $t$ cannot contain
    a private function symbol.
\end{proof}

\subsection{Examples from causation
literature}\label{sec:ex-causation}

As our accountability definition is rooted in causation, we chose to
model three examples from the literature and verify if our
accountability judgement corresponds to the results that philosophers
consider intuitive.

The first two examples originate from Lewis' seminal work on
counterfactual reasoning~\cite{Lewis1973-LEWC}
and consider different kinds of preemption.
The first, late preemption as formulated by Hitchcock~\cite[p.
526]{Hitchcock2007-HITPPA},
considers a process that would cause an event (e.g., a violation),
but is interrupted by a second processes.
An agent $A$ poisons an agent $C$, but before the poison takes effect,
$B$ shoots $C$.
Na\"{i}ve counterfactual reasoning fails because had $B$ not shot
$C$, she would still have died (from poisoning).
\begin{full}
We represent the actions as inserts into the store. Depending on the
value in the store $C$, proceeds on one out of four
different control-flow paths. The verdict function that maps the
case where the poisoning takes effect, i.e., the `poisoning' message
is received, to $A$, and the cases where a `shoot' message is received
to $B$ provides the protocol with accountability for the property that
$C$ is neither shot, nor poisoned.
\end{full}
The second, early preemption, has Suzy and Billie throwing stones at
a bottle, with Suzy going first. Billie's throw, despite being
accurate, misses the bottle if Suzy's throw shatters the bottle first.
\begin{full}
The modelling is straight-forward and captures both possible orders.
\end{full}
Again, we can capture that the message which arrives (i.e., who threw
the first stone) determines the cause the bottle shattered (see
Figure~\ref{fig:early-preemption}).
\begin{conf}
The third example is a slight modification of McLaughlin's case
about a desert traveller $C$ with a bottle of water and two
enemies~\cite[p.155]{edsjsr.10.2307.132848419250101} and also
covers late preemption.
\end{conf}

\begin{full}
The third example is a slight modification of McLaughlin's case
about a desert traveller $C$ with a bottle of water and two enemies
~\cite[p.155]{edsjsr.10.2307.132848419250101}.
One of his enemies, $A$, contaminates the water with poison, which would
result in death by poisoning, but the other enemy $B$ shoots the
bottle, resulting in death by dehydration. 
%
%
If $A$ poisons the water first and $C$ drinks it, then no matter what
$B$ does, $C$ dies of poisoning, whereas if $B$ shoots the bottle
first, then it prevents other actions and $C$ dies of dehydration.
However, if $C$ feels thirsty and drinks water beforehand, then he
will be healthy. Again, we use the store to reflect the state of the
canteen (filled, contaminated or broken).
\end{full}

\begin{full}
\begin{figure}
\centering
\begin{lstlisting}
A$\defeq$ 0 //Poisoner -- for default, no action
B$\defeq$ 0 //Bottle shooter -- dito
C$\defeq$ //Desert traveller -- victim
(in(sig_p); //if the canteen is poisoned..
 (in(sig_s); // .. and then shoot, $C$ will die of thirst
  event CanteenShoot();event Control('0','3');event Verdict()))
 (event Poisoned();  // .. if it is not shoot, he is poisoned
     (event Control('0','1'); event Verdict()) // shot could arrive later
    +(in(sig_s); event CanteenShoot(); event Control('0','2');event Verdict())))
+ // if the canteen is not poisoned, it could be shoot, or not
(in(sig_s); event CanteenShoot(); event Control('0','4');event Verdict())
+
(event Control('0','5');event Healthy(); event Verdict()) 

new p; new s; (A || B || C ||
 !(in ('c',<'corrupt',x>); event Corrupted(x); 
     !((if x='A' then out(sign(sk('A'), p)))
     ||(if x='B' then out(sign(sk('B'), s))))))
\end{lstlisting}
\caption{Thirst or Poisoning.}
\label{fig:thirst-poisoning}
\end{figure} 
\end{full}

\begin{figure}
\centering
\begin{lstlisting}
S$\defeq$ 0 //Suzy -- by default, no throw
B$\defeq$ 0 //Billy -- dito
C$\defeq$ //Bottle
    (in(a); event Shot_S(); event Shattered())
  + (in(b); event Shot_B(); event Shattered())

new a; new b; (S || B || C ||
 !(in ('c',<'corrupt',x>); event Corrupted(x); 
     !( (if x='S' then out(a)) || (if x='B' then out(b)))))
\end{lstlisting}
\caption{Early preemption.}
\label{fig:early-preemption}
\end{figure} 

It may come as a surprise that these examples can be treated without
much trouble, despite some of them being discussed for decades, and
despite our relatively simple notion of causation.
The reason is that our calculus 
forces the modeller to describe the control-flow,
which can be argued to be an essential part of the correct modelling of
puzzling cases based on preemption.
The relation $r_c$  captures the intuition that 
counterfactuals can be ignored if they do not represent the actual
course of events, e.g., the control-flow.
We do not claim that our notion of causation solves the matter of
causation --- it is not clear what control-flow means outside of
programming languages, for example ---
but that it, expectedly, works well on examples that involve preemption.

\end{full}
\begin{full}
\subsection{Relation to Bruni et\,al.}\label{sec:giustolisi}

According to Bruni, Giustolisi and Sch{\"{u}}rmann~\cite{DBLP:conf/isw/BruniGS17}
a test is providing accountability, if
four conditions hold: verifiability-soundness,
verifiability-completeness, test-soundness and test-completeness.
These can be compared to our verification conditions for $r_w$ (other
relations cannot be expressed). In the simplified case where
verdicts can only contain singletons (joint accountability cannot be
expressed),
verifiability-soundness is equivalent to 
the $\impliedby$-direction of $\Verifiability$,
verifiability-completeness is equivalent to $\Minimality$ (the only strict subset of a singleton is the empty set),
test-completeness is equivalent to $\Uniqueness$,
and
test-soundness is implied by $\Verifiability$ but not vice
versa (not even the $\implies$-direction) .
We remark that there is no equivalent to $\SufficiencyB$. Indeed, 
a protocol that does not permit violation ($\varphi$ is always
true), could fulfill all four conditions and blame still blame
a party against any intuition, hence this condition is truly missing.

%
Furthermore, test-soundness is too weak, i.e., it is possible to have
violations that go undetected (contrary to the intuition behind
verifiability).

\end{full}
\end{document}